%% file: paper.tex
\RequirePackage{afterpackage}

\AfterPackage{amsthm}{
  \RequirePackage{hyperref,cleveref}
}

\documentclass[10pt]{article}
\usepackage[utf8]{inputenc}
\usepackage{amssymb}
\usepackage{custom-macros}
\title{\textbf{Optimal (degree+1)-Coloring in Congested Clique}\thanks{A preliminary version of this paper appeared in \emph{Proceedings of the 50th International Colloquium on Automata, Languages, and Programming (ICALP)}, pages 45:1--45:20, 2023.}}

\author{\textbf{Sam Coy}\thanks{E-mail: S.Coy@warwick.ac.uk. Department of Computer Science, University of Warwick, UK.
Research supported in part by the Centre for Discrete Mathematics and its Applications (DIMAP),
by an EPSRC studentship, and
by the Simons Foundation Award No. 663281 granted to the Institute of Mathematics of the Polish Academy of Sciences for the years 2021--2023.}
\\
University of Warwick
\and
\textbf{Artur Czumaj}\thanks{E-mail: A.Czumaj@warwick.ac.uk. Department of Computer Science and Centre for Discrete Mathematics and its Applications (DIMAP), University of Warwick, UK. Research supported in part by the Centre for Discrete Mathematics and its Applications, by EPSRC award EP/V01305X/1, by a Weizmann-UK Making Connections Grant, by an IBM Award, and by the Simons Foundation Award No. 663281 granted to the Institute of Mathematics of the Polish Academy of Sciences for the years 2021--2023.}
\\
University of Warwick
\and
\textbf{Peter Davies}\thanks{E-mail: Peter.W.Davies@durham.ac.uk.}
\\
Durham University
\and
\textbf{Gopinath Mishra}\thanks{E-mail: Gopinath.Mishra@warwick.ac.uk. Department of Computer Science, University of Warwick, UK. Research supported in part by the Centre for Discrete Mathematics and its Applications (DIMAP), by EPSRC award EP/V01305X/1, and by the Simons Foundation Award No. 663281 granted to the Institute of Mathematics of the Polish Academy of Sciences for the years 2021--2023.}
\\
University of Warwick
}

\author{
\begin{tabular}[t]{c c c c}
\textbf{Sam Coy}\thanks{E-mail: S.Coy@warwick.ac.uk. Department of Computer Science, University of Warwick, UK.
Research supported in part by the Centre for Discrete Mathematics and its Applications (DIMAP),
by an EPSRC studentship, and
by the Simons Foundation Award No 663281 granted to the Institute of Mathematics of the Polish Academy of Sciences for the years 2021--2023.} &
\textbf{Artur Czumaj}\thanks{E-mail: A.Czumaj@warwick.ac.uk. Department of Computer Science and Centre for Discrete Mathematics and its Applications (DIMAP), University of Warwick, UK. Research supported in part by the Centre for Discrete Mathematics and its Applications (DIMAP), by EPSRC award EP/V01305X/1, by a Weizmann-UK Making Connections Grant, by an IBM Award, and by the Simons Foundation Award No. 663281 granted to the Institute of Mathematics of the Polish Academy of Sciences for the years 2021--2023.} &
\textbf{Peter Davies}\thanks{E-mail: Peter.W.Davies@durham.ac.uk. Department of Computer Science, Durham University, UK.}  &
\textbf{Gopinath Mishra}\thanks{E-mail: Gopinath@imsc.res.in. Theoretical Computer Science Group, The Institute of Mathematical Sciences, HBNI, Chennai, India. Research supported in part  by EPSRC award EP/V01305X/1, and by the Simons Foundation Award No 663281 granted to the Institute of Mathematics of the Polish Academy of Sciences for the years 2021-2023.} \\
\small University of Warwick &
\small University of Warwick &
\small Durham University &
\small  The Institute of Mathematical Sciences
\end{tabular}
}

\date{}

\begin{document}


\maketitle
\begin{abstract}
We consider the distributed complexity of the (\textsf{degree+1})-list coloring problem, in which each node $u$ of degree $d(u)$ is assigned a palette of $d(u)+1$ colors, and the goal is to find a proper coloring using these color palettes. The (\textsf{degree+1})-list coloring problem is a natural generalization of the classical $(\Delta+1)$-coloring and $(\Delta+1)$-list coloring problems, both being benchmark problems extensively studied in distributed and parallel computing.

In this paper we settle the complexity of the (\textsf{degree+1})-list coloring problem in the Congested Clique model by showing that it can be solved deterministically in a constant number of rounds.
\end{abstract}


\input{parts/preliminaries}

\input{parts/algorithm2}

\paragraph*{Acknowledgment:} We would like to thank the anonymous reviewers of SICOMP for their insightful suggestions, which improved the presentation of the paper and helped us obtain a rigorous proof of \Cref{lem:colorstep}.


\bibliographystyle{alpha}
\bibliography{references}




\end{document}

%% file: parts/preliminaries.tex
\section{Introduction}
\label{sec:Introduction}

Graph coloring problems are among the most extensively studied problems in the area of distributed graph algorithms.
In the distributed graph coloring problem, we are given an undirected graph $G = (V,E)$ and the goal is to properly color the nodes of $G$ such that no edge in $E$ is monochromatic. In the distributed setting, the nodes of $G$ correspond to devices that interact by exchanging messages throughout some underlying communication network such that the nodes communicate with each other in synchronous rounds by exchanging messages over the edges in the network. Initially, the nodes do not know anything about $G$ (except possibly for some global parameters, e.g., the number of nodes $n$ or the maximum degree $\Delta$). At the end of computation, each node $v \in V$ should output its color (from a given domain) in the computed coloring. The \emph{time} or \emph{round complexity} of a distributed algorithm is the total number of rounds until all nodes terminate.

If adjacent nodes in $G$ can exchange arbitrarily large messages in each communication round (and hence the underlying communication network is equal to the input graph $G$), this distributed model is known as the \LOCAL model \cite{Linial92}, and if messages are restricted to $O(\log n)$ bits per edge (limited bandwidth) in each round, the model is known as the \CONGEST model \cite{Peleg00}. If we allow all-to-all communication (i.e., the underlying network is a complete graph and thus the communication is independent of the input graph $G$) using messages of size $O(\log n)$ bits then the model is known as the \CONGESTEDC model \cite{LPPP05}.

The most fundamental graph coloring problem in distributed computing (studied already in the seminal paper by Linial \cite{Linial92} that introduced the \LOCAL model) is \emph{$(\Delta+1)$-coloring}: assuming that the input graph $G$ is of maximum degree $\Delta$, the objective is to properly color nodes of $G$ using $\Delta+1$ colors from $\{1,2,\dots,\Delta+1\}$. The $(\Delta+1)$-coloring problem can be easily solved by a sequential greedy algorithm, but the interaction between local and global aspects of graph coloring create some non-trivial challenges in a distributed setting.
The problem has been used as a benchmark to study distributed symmetry breaking in graphs, and it is at the very core of the area of distributed graph algorithms.
\emph{$(\Delta+1)$-list coloring} is a natural generalization of $(\Delta+1)$-coloring: each node has an arbitrary palette of $\Delta+1$ colors, and the goal is to compute a legal coloring in which each node is assigned a color from its own palette.
A further generalization is the \emph{(\textsf{degree+1})-list coloring (D1LC)} problem, which is the same as the $(\Delta+1)$-list coloring problem except that the size of each node $v$'s palette is $d(v) + 1$, which might be much smaller than $\Delta+1$.
These three problems always have a legal coloring (easily found sequentially using a greedy approach), and the main challenge in the distributed setting is to find the required coloring in as few rounds as possible.

These three graph coloring problems have been studied extensively in distributed computing, though $(\Delta+1)$-coloring, as the simplest, has attracted most attention. However, one can also argue that (\textsf{degree+1})-list coloring, as the most versatile, is more algorithmically fundamental than $(\Delta+1)$-coloring. For example, given a partial solution to a $(\Delta+1)$-coloring problem, the remaining coloring problem on the uncolored nodes is an instance of the (\textsf{degree+1})-list coloring problem. The (\textsf{degree+1})-list coloring problem is self-reducible: after computing a partial solution to a (\textsf{degree+1})-list coloring problem, the remaining problem is still a (\textsf{degree+1})-list coloring problem. It also naturally appears as a subproblem in more constrained coloring problems: for example, it has been used as a subroutine in distributed $\Delta$-coloring algorithms (see, e.g., \cite{FHM23}), in efficient $(\Delta+1)$-coloring and edge-coloring algorithms (see, e.g., \cite{Kuhn20}), and in other graph coloring applications (see, e.g.~\cite{BE19}).

Following an increasing interest in the distributed computing community for (\textsf{degree+1})-list coloring, it is natural to formulate a central challenge relating it to $(\Delta+1)$-coloring:

\medskip
\centerline{%
\parbox{5.0in}{
\begin{mdframed}[hidealllines=true,backgroundcolor=gray!15]\begin{quote}
    \emph{Can we solve the (\textsf{degree+1})-list coloring problem in asymptotically the same round complexity as the simpler $(\Delta+1)$-coloring problem?}
\end{quote}\end{mdframed}
}%
}
\medskip

\comments{This challenge has been elusive for many years and only recently was the affirmative answer given for \emph{randomized} algorithms in \LOCAL and \CONGEST. In particular, the (\textsf{degree+1})-list coloring problem admits a randomized $\widetilde{O}(\log^2\log n)$-round distributed algorithm in the \LOCAL model \cite{HKNT22,GG23} matching the state-of-the-art complexity for the $(\Delta+1)$-coloring problem \cite{CLP20,GG23}.~\footnote{$\widetilde{O}(f)$ hides a  polynomial term in $\log f$.} This has been later extended to the \CONGEST model with a complexity of $\widetilde{O}(\log^3\log n)$ rounds \cite{HNT22,GGHIR23}, matching the state-of-the-art complexity for the $(\Delta+1)$-coloring problem in \CONGEST \cite{HKMT21,GK21}.}

The main contribution of our paper is a complete resolution of this challenge in the \CONGESTEDC model, and in fact, even for deterministic algorithms. We settle the complexity of the (\textsf{degree+1})-list coloring problem in \CONGESTEDC by showing that it can be solved deterministically in a constant number of rounds.
\hide{\Anote{Didn't we also want to argue that we can solve \DILC in a constant number of rounds on an \MPC with $O(n)$ local space and $O(1+\frac{m}{n})$ machines? (Or we may leave it for now, but put it in the arxiv/full version.)}}


\begin{mdframed}[hidealllines=true,backgroundcolor=gray!25]\vspace{-8pt}
\begin{theorem}
\label{thm:D1LCcolor}
There is a deterministic \CONGESTEDC algorithm which finds a (\textsf{degree+1})-list coloring of any graph in a constant number of rounds.
\end{theorem}
\end{mdframed}

\hide{
\Anote{I guess we will defer this sentence and \Cref{thm:DeltaCcolor} to the next time \dots}%
We will also present an application of \Cref{thm:D1LCcolor} to the distributed Brooks' theorem in  \CONGESTEDC.

\begin{theorem}
\label{thm:DeltaCcolor}
There is a randomized \CONGESTEDC algorithm that in $O(1)$ rounds finds a $\Delta$-coloring of any graph (with $\Delta\ge 3$ that does not contain a $(\Delta+1)$-clique), succeeding with high probability in $n$.
\end{theorem}}


\subsection{Background and Related Works}
\label{subsec:related-works}

The distributed graph coloring problems have been extensively studied in the last three decades, starting with a seminal paper by Linial \cite{Linial92} that introduced the \LOCAL model and originated the area of local graph algorithms. Since the $(\Delta+1)$-coloring 
problem can be solved by a simple sequential greedy algorithm, but it is challenging to be solved efficiently in distributed (and parallel) setting, the $(\Delta+1)$-coloring problem became a benchmark problem for distributed computing and a significant amount of research has been devoted to the study of these problems in all main distributed models: \LOCAL, \CONGEST, and \CONGESTEDC. The monograph 
\cite{BE13} gives a comprehensive description of many of the earlier results.

It is known from research on parallel algorithms that $(\Delta+1)$-coloring can be computed in $O(\log n)$ rounds by randomized algorithms in the \LOCAL model \cite{ABI86,Luby86}. Linial \cite{Linial92} observed that for smaller values of $\Delta$, one can do better: he showed that it is possible to deterministically color arbitrary graphs of maximum degree $\Delta$ with $O(\Delta^2)$ colors in $O(\log^*n)$ rounds; this can be easily extended to obtain a deterministic \LOCAL algorithm for $(\Delta+1)$-coloring that runs in $O(\Delta^2 + \log^*n)$ rounds, and thus in bounded degree graphs, a $(\Delta+1)$-coloring can be computed in $O(\log^*n)$ rounds. These results have since been improved for general values of $\Delta$: \comments{the current state-of-the-art for the $(\Delta+1)$-coloring problem in \LOCAL is $\widetilde{O}(\log^2\log n)$ rounds for randomized algorithms \cite{CLP20,GG23} and $\widetilde{O}(\log^2 n)$ for deterministic algorithms \cite{GG23}}. Furthermore, the fastest algorithms mentioned above can be modified to work also for the more general $(\Delta+1)$-list coloring problem in the \LOCAL model. (In fact, many of those algorithms critically rely on this problem as a subroutine.)

For the \CONGEST model, the parallel algorithms mentioned above \cite{ABI86,Luby86} can be implemented in the \CONGEST model to obtain randomized algorithms for both the $(\Delta+1)$-coloring and $(\Delta+1)$-list coloring problems that run in $O(\log n)$ rounds. \comments{Only recently this bound has been improved for all values of $\Delta$: The result in a seminal paper by Halld{\'o}rsson \etal \cite{HKMT21} when combined with the result of Ghaffari \etal \cite{GGHIR23} implies a randomized \CONGEST algorithm that solves the $(\Delta+1)$-coloring and $(\Delta+1)$-list coloring problems in $\widetilde{O}(\log^3\log n)$ rounds. For deterministic computation, the best known
algorithm \cite{GK21} runs in $\widetilde{O}(\log^2\Delta \cdot \log n)$ rounds.}

As for the lower bounds, one of the first results in distributed computing was a lower bound in \LOCAL of $\Omega(\log^*n)$ rounds for computing an $O(1)$-coloring of a graph of maximum degree $\Delta = 2$, shown by Linial \cite{Linial92} for deterministic algorithms, and by Naor \cite{Naor91} for randomized ones. Stubbornly, the $\Omega(\log^*n)$ rounds is still the best known lower bound for the $(\Delta+1)$-coloring problem in \LOCAL and \CONGEST.

We can do better for the \CONGESTEDC model. After years of gradual improvements, Parter \cite{Parter18} exploited the \LOCAL shattering approach from \cite{CLP20} to give the first sublogarithmic-time randomized $(\Delta+1)$-coloring algorithm for \CONGESTEDC, which runs in $O(\log\log\Delta \cdot \log^*\Delta)$ rounds. This bound has been later improved by Parter and Su \cite{PS18} to $O(\log^*\Delta)$ rounds. Finally, Chang \etal \cite{CFGUZ19} settled the randomized complexity of $(\Delta+1)$-coloring (and also for $(\Delta+1)$-list coloring) and obtained a randomized \CONGESTEDC algorithm that runs in a constant number of rounds. This result has been later simplified and turned into a deterministic constant-round \CONGESTEDC algorithm by Czumaj \etal \cite{CDP20}.

\paragraph{(\textsf{degree+1})-list coloring (\DILC).}
%
The \DILC problem in distributed setting has been studied both on its own, and also as a tool in designing distributed algorithms for other coloring problems, like $(\Delta+1)$-coloring, $(\Delta+1)$-list coloring, and $\Delta$-coloring. The problem is not easier than the $(\Delta+1)$-coloring and the $(\Delta + 1)$-list coloring problems, and the difficulty of dealing with vertices having color palettes of significantly different sizes makes the problem more challenging. As the result, until very recently the obtained complexity bounds have been significantly weaker than the bounds for the $(\Delta+1)$-coloring problem, see, e.g., \cite{BKM20,FHK16,Kuhn20}.
\comments{This changed in 2022, when in a  breakthrough Halld{\'o}rsson \etal \cite{HKNT22} gave a randomized $O(\log^3 \log n)$-round distributed algorithm for \DILC in the \LOCAL distributed model. When this result is combined with the result by Ghaffari and Grunau \cite{GG23}, it implies an $\widetilde{O}(\log^2 \log n)$-round \LOCAL algorithm for \DILC. Observe that this bound matches the state-of-the-art complexity for the (easier) $(\Delta+1)$-coloring problem \cite{CLP20,GG23}.
The algorithm for \DILC in the \LOCAL model was later extended to the \CONGEST model by Halld{\'o}rsson \etal \cite{HNT22}, who designed a randomized \CONGEST algorithm for \DILC that runs in $O(\log^5 \log n)$ rounds. When this result is combined with the network decomposition result by Ghaffari \etal \cite{GGHIR23} implies $\widetilde{O}(\log^3 \log n)$-round \CONGEST algorithm for \DILC. Similarly as for the \LOCAL model, this bound matches the state-of-the-art complexity for the $(\Delta+1)$-coloring problem in \CONGEST \cite{HKMT21,GGHIR23}.}

Specifically for the \CONGESTEDC model, the only earlier \DILC result we are aware of is by Bamberger \etal \cite{BKM20}, who extended their own \CONGEST algorithm for the problem to obtain a deterministic \DILC algorithm requiring $O(\log\Delta\log\log\Delta)$ rounds in \CONGESTEDC. \comments{However, the randomized state-of-the-art \DILC bound in the \CONGESTEDC model follows from the aforementioned $\widetilde{O}(\log^3\log n)$-round \CONGEST algorithm \cite{HNT22,GGHIR23}, which works directly in \CONGESTEDC}. This should be compared with the state-of-the-art $O(1)$-round \CONGESTEDC algorithms for $(\Delta+1)$-coloring 
\cite{CFGUZ19,CDP20}.

\paragraph{Recent work in \DILC on \MPC.}
Various coloring problems have been also studied in a related model of parallel computation, the so-called \emph{Massively Parallel Computation} (\MPC) model. The \MPC model, introduced by Karloff \etal \cite{KSV10}, is now a standard theoretical model for parallel algorithms.
%
The \MPC model with $O(n)$ local space and $n$ machines is essentially equivalent to the \CONGESTEDC model (see, e.g., \cite{BDH18,HP15}), and this implies that many \MPC algorithms can be easily transferred to the \CONGESTEDC model. (However, this relationship
requires that the local space of \MPC is $O(n)$, not more.)

Both the $(\Delta+1)$-coloring and $(\Delta+1)$-list coloring problems have been studied in \MPC extensively (see, e.g., \cite{BKM20,CDP20} for linear local space \MPC and \cite{BKM20,CFGUZ19,CDP21a} for sublinear local space \MPC).
We are aware only of a few works for the \DILC problem on \MPC, see \cite{BKM20,CCDM23,HKNT22}. The work most relevant to our paper is the result of Halld{\'o}rsson \etal \cite{HKNT22}.
They give a constant-round \MPC algorithm assuming the local \MPC space is \emph{slightly superlinear}, i.e., $\Omega(n \log^4n)$ \cite[Corollary~2]{HKNT22}. This result relies on the palette sparsification approach due to Alon and Assadi \cite{AlonA20} (see also \cite{ACK19}) to the \DILC problem, which reduces the problem to a sparse instance of size $O(n \log^4n)$; hence, on an \MPC with $\Omega(n \log^4n)$ local space one can put the entire graph on a single \MPC machine and then solve the problem in a single round.
\hide{
Both the $(\Delta+1)$-coloring and $(\Delta+1)$-list coloring problems have been studied in the \MPC extensively (see, e.g., \cite{BKM20,CDP20} for the results for linear local space \MPC and \cite{BKM20,CFGUZ19,CDP21a} for the results for sublinear local space \MPC).
As for the \DILC problem on \MPC, we are aware of only three recent works \cite{BKM20,CCDM23,HKNT22}. In the setting of sublinear local space \MPC, very recently Coy \etal \cite{CCDM23} obtained a randomized $O(\log\log\log n)$-round \MPC algorithm for \DILC, improving upon an earlier bound of $O(\log^2 \Delta+\log n)$ \cite{BKM20}. Bamberger \etal \cite{BKM20} give also (\textsf{degree+1})-list coloring algorithms requiring $O(\log^2 \Delta)$ rounds in $O(n)$ local space \MPC (the product of the number of machines times the local space is $\widetilde{O}(n+m)$; if the number of machines is allowed to be as large as $\Theta(n)$, as in \CONGESTEDC, then the number of rounds reduces to $O(\log\Delta\log\log\Delta)$, as stated above). 
The framework of Halld{\'o}rsson \etal \cite{HKNT22} developed for the randomized $O(\log^3\log n)$-round distributed algorithm for \DILC in \LOCAL can be easily incorporated into a constant-round \MPC algorithm assuming the local \MPC space is \emph{slightly superlinear}, i.e., $\Omega(n \log^4n)$ \cite[Corollary~2]{HKNT22}. This result relies on the color palette sparsification approach due to Alon and Assadi \cite{AlonA20} (see also \cite{ACK19}) to the \DILC problem, which in the \MPC setting allows to reduce the problem to a sparse instance of size $O(n \log^4n)$; hence, by having $O(n \log^4n)$ local space one can put the entire graph on a single \MPC machine and then solve the problem in a single round.
}%
\hide{
Both the $(\Delta+1)$-coloring and $(\Delta+1)$-list coloring problems have been studied in the \MPC extensively (see, e.g., \cite{BKM20,CDP20} for the results for linear local space \MPC and \cite{BKM20,CFGUZ19,CDP21a} for the results for sublinear local space \MPC).
As for the \DILC problem on \MPC, we are aware of only three recent works \cite{BKM20,CCDM23,HKNT22}. Bamberger \etal \cite{BKM20}, who studied the problem in the \CONGEST model, give also (\textsf{degree+1})-list coloring algorithms requiring $O(\log^2 \Delta)$ rounds in $O(n)$ local space \MPC and $O(\log^2 \Delta+\log n)$ in sublinear local space \MPC (in both cases the product of the number of machines times the local space is $\widetilde{O}(n+m)$; if the number of machines is allowed to be as large as $\Theta(n)$, as in \CONGESTEDC, then the number of rounds reduces to $O(\log\Delta\log\log\Delta)$, as stated above). The bound for sublinear local space \MPC has been very recently improved by Coy \etal \cite{CCDM23}, who obtain a randomized $O(\log\log\log n)$-round \MPC algorithm for \DILC. 
The framework of Halld{\'o}rsson \etal \cite{HKNT22} developed for the randomized $O(\log^3\log n)$-round distributed algorithm for \DILC in \LOCAL can be easily incorporated into a constant-round \MPC algorithm assuming the local \MPC space is \emph{slightly superlinear}, i.e., $\Omega(n \log^4n)$ \cite[Corollary~2]{HKNT22}. This result relies on the color palette sparsification approach due to Alon and Assadi \cite{AlonA20} (see also \cite{ACK19}) to the \DILC problem, which in the \MPC setting allows to reduce the problem to a sparse instance of size $O(n \log^4 n)$; hence, by having $O(n \log^4n)$ local space one can put the entire graph on a single \MPC machine and then solve the problem in a single round.
}%
Given the similarity of \CONGESTEDC and the \MPC model with \emph{linear local space}, one could hope that the use of ``slightly superlinear'' \MPC local space in \cite{HKNT22} can be overcome and the approach can allow the problem to be solved in linear local space, resulting in a \CONGESTEDC algorithm with a similar performance. However, these  sparsification techniques are naturally limited to producing instances of $O(n poly(\log n))$ size, and these instances  still appear difficult to color in $O(1)$ rounds in \CONGESTEDC, and it is not clear how to extend that approach.

Further, we have recently seen a similar situation in $(\Delta+1)$-coloring. The palette sparsification by Assadi \etal \cite{ACK19} trivially implies a constant-round \MPC algorithm for $(\Delta+1)$-coloring with local space $\Omega(n \log^2 n)$, but does not give a constant-round algorithm for $(\Delta+1)$-coloring in \CONGESTEDC. Only by using a fundamentally different approach Chang \etal \cite{CFGUZ19} and then (deterministically) Czumaj \etal \cite{CDP21c} obtained constant-round $(\Delta+1)$-coloring algorithms in \CONGESTEDC.
Hence, despite having a constant-round algorithm for \DILC in \MPC with local space $\Omega(n \log^4 n)$, possibly a different approach than palette sparsification is needed to achieve a similar performance for \DILC in \CONGESTEDC.
%

\paragraph{Derandomization tools for distributed coloring algorithms.}
In our paper we rely on a recently developed general scheme for derandomization in the \CONGESTEDC model (and used also extensively in the \MPC model) by combining the methods of bounded independence with efficient computation of conditional expectations. This method was first applied by Censor-Hillel \etal \cite{CPS20}, and has since been used in several other works for graph coloring problems, (see, e.g., \cite{CDP20,Parter18}), and for other problems in \CONGESTEDC and \MPC. 

The underlying idea begins with the design of a randomized algorithm using random choices with only \emph{limited independence}, e.g., $O(1)$-wise-independence. 
Then, each round of the randomized algorithm can be simulated by giving all nodes a shared random seed of $O(\log n)$ bits. Next, the nodes deterministically compute a seed which is at least as good as a random seed is in expectation. This is done by using an appropriate estimation of the local quality of a seed, which can be aggregated into a global measure of the quality of the seed. Combining this with the techniques of conditional expectation, pessimistic estimators, and bounded independence, this allows selection of the bits of the seed ``batch-by-batch,'' where each batch consists of $O(\log n)$ bits. Once all bits of the seed are computed, we can used it to simulate the random choices of that round, as it would have been performed by a randomized algorithm.
A more detailed explanation of this approach is given in \Cref{subsec:mpc_derandomization}.

\subsection{Technical Overview}
\label{subsec:techniques}

The core part of our constant-round deterministic \CONGESTEDC \DILC algorithm (\textsc{BucketColor}, \Cref{alg:bucketcolor}) does \emph{not} follow the route of recent \DILC algorithms for \LOCAL and \CONGEST due to Halld{\'o}rsson \etal \cite{HKNT22,HNT22}. Instead, we employ fundamentally different techniques, building on a graph partitioning approach developed in a series of papers \cite{Parter18, PS18, CFGUZ19, CDP20}. Chang \etal \cite{CFGUZ19} showed that their version of this partitioning procedure could yield an $O(1)$-round \CONGESTEDC algorithm for $(\Delta+1)$-list coloring, and Czumaj, Davies, and Parter~\cite{CDP20} then showed that a derandomized version gives a deterministic algorithm with the same complexity. For the purposes of providing the randomized intuition behind the process, the procedure of \cite{CFGUZ19} would already suffice, but we will use the terminology of \cite{CDP20} for easier comparison in our subsequent derandomization.

 This partitioning-based $(\Delta+1)$-list coloring algorithm works by partitioning nodes and colors into $\Delta^\eps$ \emph{buckets}, for a 
small constant $\eps$. (This partitioning is initially at random, but then it is derandomized in \cite{CDP20} using the method of conditional expectations). The nodes are distributed among all buckets, and the colors are distributed among all but one (the \emph{``leftover'' bucket}), which is left without colors and is set aside to be colored later. Then, each node's palette is restricted to only the colors assigned to its bucket (except those in the leftover bucket, whose palettes are not restricted). This ensures that nodes in different buckets have entirely disjoint palettes, and therefore edges between different buckets can be removed from the graph, since they would never cause a coloring conflict. One important property is that nodes still have sufficient colors when their palettes were restricted in this way. This is achieved in \cite{CFGUZ19} using two main arguments: firstly, the fact that colors are distributed among one fewer buckets than nodes provided enough `slack' to ensure that with reasonably high probability, a node would receive more colors than neighbors in its bucket. Secondly, the few nodes that do \emph{not} satisfy this property induce a small graph (of size $O(n)$), and therefore can be collected onto a single network node and color separately.

Using the approach sketched above, a $(\Delta+1)$-list coloring instance can be reduced into multiple smaller $(\Delta+1)$-list coloring instances (i.e., on fewer nodes and with a new, lower maximum degree) that are \emph{independent} (since they had disjoint palettes), and so can be solved in parallel without risking coloring conflicts. The final part of the analysis of \cite{CFGUZ19} was to show that, after recursively performing this bucketing process $O(1)$ times, these instances are of $O(n)$ size and therefore they could be collected to individual nodes and solved in a constant number of rounds in \CONGESTEDC.

There are major barriers to extending this approach to (\textsf{degree+1})-list coloring. Crucially, it required the number of buckets to be dependent on $\Delta$, and all nodes' palette sizes to be at least $\Delta$. Dividing nodes among too few buckets would cause the induced graphs to be too large, and the algorithm would not terminate in $O(1)$ rounds; using too many buckets would fail to provide nodes with sufficient colors in their bucket. In the \DILC problem, we no longer have a uniform bound on palette size, so it is unclear how to perform this bucketing.

Our first major conceptual change is that, rather than simply partitioning among a number (dependent on $\Delta$) of equivalent buckets, we instead use a tree-structured hierarchy of buckets\footnote{The full recursive structure of the algorithms of \cite{CFGUZ19,CDP20} could also be viewed as a tree, but the difference is that there, at any stage, the nodes are in buckets of the same level, and the final coloring of all nodes happens at the leaf level, so the analysis need only consider a `flat' collection of buckets. In our algorithm, nodes are placed, and eventually colored, at heterogenous levels of the tree, so we must consider the full tree structure.}, with the $O(\log\log \Delta)$ levels in the hierarchy corresponding to doubly-exponentially increasing degree ranges (see Figure \ref{fig:bucket}). Nodes with degree $d(v)$ will be mapped to a bucket in a level containing (very approximately) $d(v)^{0.7}$ buckets. Colors will be mapped to a top-level (leaf) bucket, but will also be assigned to every bucket on the leaf-root path in the bucket tree (see Figure \ref{fig:part}). We can therefore discard all edges between different buckets that do not have an ancestor-descendant relationship, since these buckets will have disjoint palettes.

This change allows nodes to be bucketed correctly according to their own degree. However, it introduces several new difficulties:

\begin{itemize}

\item We no longer get a good bound on the number of lower-degree neighbors of a node $v$ that may share colors with it. We can only hope to prove that $v$ receives enough colors relative to its higher (or same)-degree neighbors.

\item The technique of having a leftover bucket which is not assigned colors no longer works to provide slack (nor even makes sense - we would need a leftover bucket at every level, but, for example, level $0$ only contains one bucket).
\end{itemize}

\begin{figure}
\centering
\begin{minipage}{.55\textwidth}
  \centering
  \includegraphics[width=.95\linewidth]{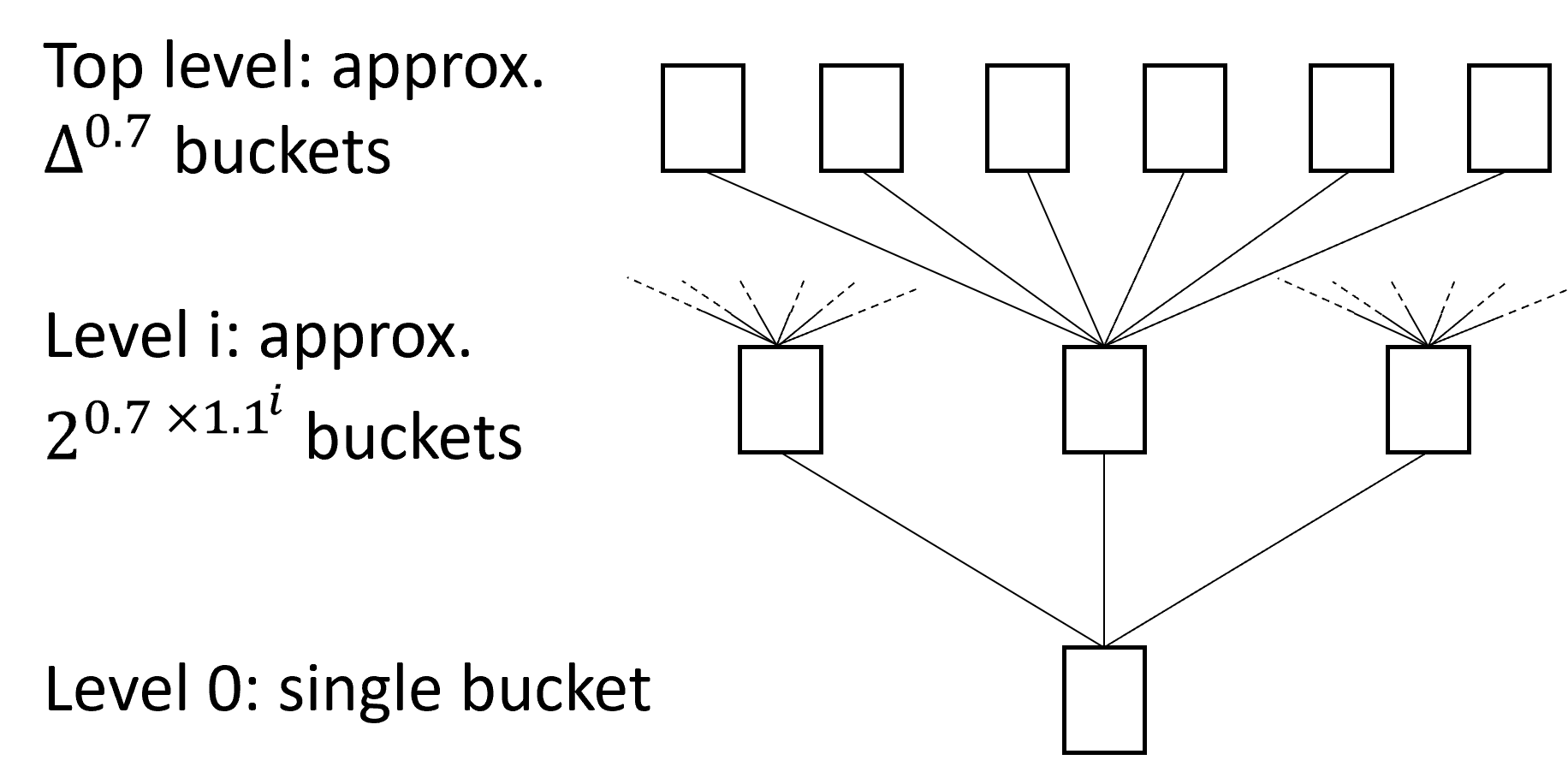}
  \captionof{figure}{Bucket structure}
  \label{fig:bucket}
\end{minipage}%
\begin{minipage}{.42\textwidth}
  \centering
  \includegraphics[width=.95\linewidth]{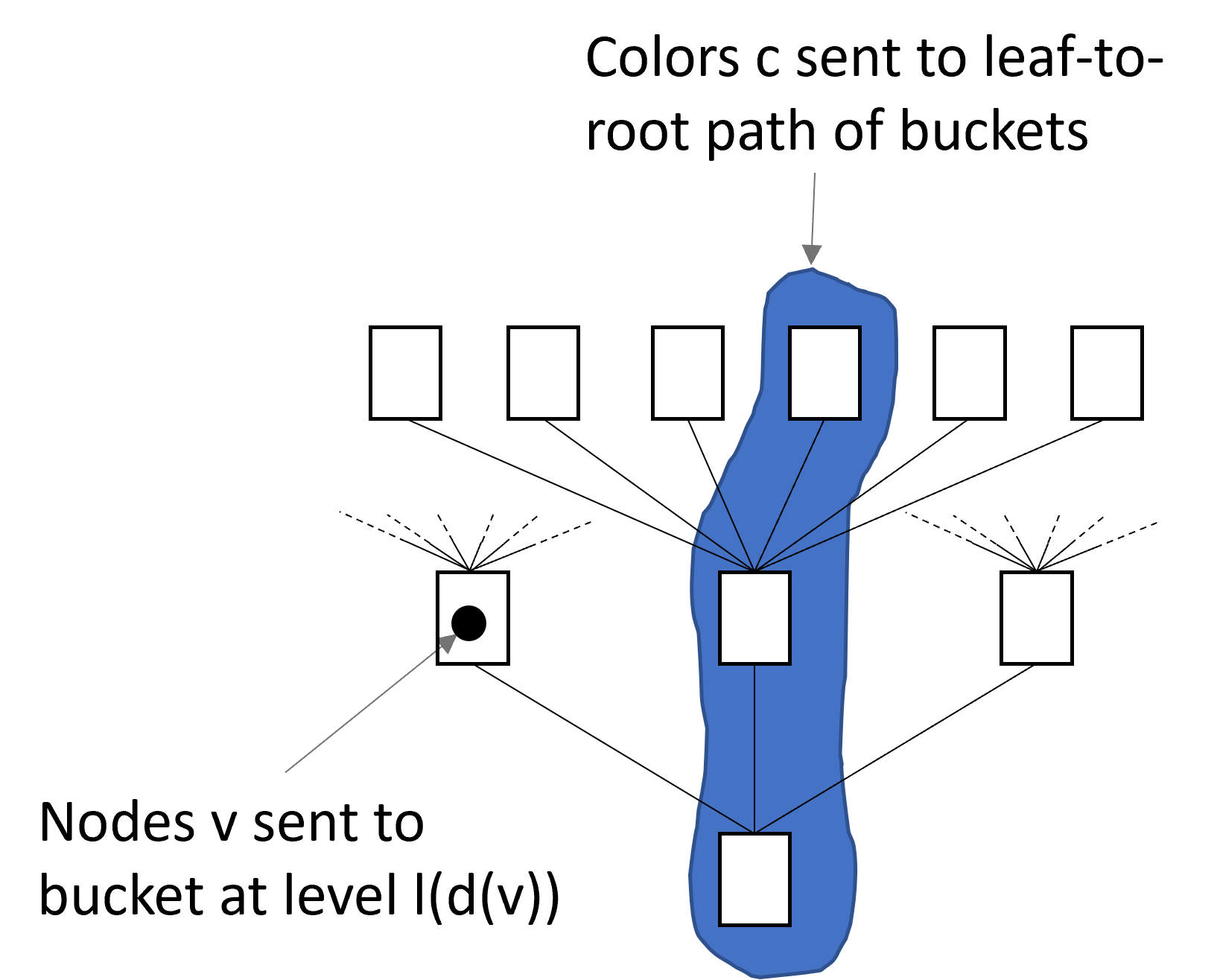}
  \captionof{figure}{Partitioning nodes and colors}
  \label{fig:part}
\end{minipage}
\end{figure}

In order to cope with these challenges, firstly, we employ the observation that if we were to greedily color in non-increasing order of degree, we would require nodes to have a palette size of $d^+(v) + 1$ (where $d^+(v)$ is the number of $v$'s neighbors of equal or higher degree), rather that $d(v)+1$ (since $d^+(v)$ of $v$'s neighbors will have been colored at the point $v$ is considered). Therefore, we argue that we can still show that the graph is colorable even though our bucketing procedure may leave nodes with many more lower-degree neighbors than palette colors. (It is not necessarily clear how to find such a coloring in a parallel fashion, but in our analysis, we will be able to address this issue.)

This observation also helps us with the problem of generating slack without a leftover bucket. We show that, since lower-degree neighbors are now effectively providing slack, only nodes with very few lower-degree neighbors may not receive enough colors (relative to higher-degree neighbors) in their bucket. It transpires that we can generate slack for these nodes prior to \textsc{BucketColor} via derandomizations of fairly standard procedures (\textsc{ColorTrial}, \Cref{alg:colortrialrand}, and \textsc{SubSample}, \Cref{alg:subsamplerand}). The randomized bases for all these procedures would inevitably result in some nodes failing to meet the necessary properties for the next stage. To overcome this, we derandomize all of these procedures, using the method of conditional expectations. As well as making the algorithm deterministic, this has the important property of ensuring that \emph{failed} nodes  form an $O(n)$-size induced graph, which can be easily dealt with later.

Having solved the problem of slack for the bucketing process (by showing that nodes have received palettes of size at least $d^+(v) + 1$  within their buckets), it remains to find a parallel analog to greedily coloring in non-increasing order of degree. Our approach here is to repeatedly move all nodes from their current bucket to a child of that bucket in the bucket tree (which further restricts their neighborhood and available palette). We show that, by correct choice of bucket and order of node consideration, we will always be able to find child buckets such that each node still has palette size at least $d^+(v) + 1$ according to the new bucket assignment. We also show that, after $O(1)$ iterations of this process, nodes only have one palette color in their bucket, and zero higher-degree neighbors. Then, all nodes can safely color themselves this palette color, and the coloring is complete.

The overall structure of the main algorithm \textsc{Color} (\Cref{alg:mainrand}) is more complicated, since \textsc{SubSample} produces a graph $G'$ of leftover nodes that are deferred to be colored later. We recursively run \Cref{alg:mainrand} on this graph $G'$, and show that it is sufficiently smaller than the original input graph so that after $O(1)$ recursive calls, the remaining graph has size $O(n)$ and can be collected and solved on a single network node.

If we combine all these tools together then we will be able to obtain a randomized \CONGESTEDC algorithm that finds a (\textsf{degree+1})-list coloring of any graph in $O(1)$ rounds. Using the \emph{method of conditional expectations} with \emph{bounded-independence hash function} (see \Cref{subsec:mpc_derandomization}-\ref{subsec:method_of_conditional_expectations}), each randomized step of our algorithm can be derandomized.

%% file: parts/algorithm2.tex
\section{Preliminaries}
\label{sec:preliminaries}

The main model considered in this paper is \CONGESTEDC, as introduced by Lotker \etal \cite{LPPP05}. It is a variant of \CONGEST, in which nodes can send a message of size $O(\log n)$ to each neighboring node in the graph in each communication round: the difference is that \CONGESTEDC allows all-to-all communication, and hence the underlying communication network is a complete graph on the nodes $V$. In particular, this allows the communication to be performed between all pairs of nodes rather than being restricted to the edges of the input graph. \CONGESTEDC has been introduced as a theoretical model to study overlay networks: an abstraction that separates the problems emerging from the topology of the communication network from the problems emerging from the structure of the problem at hand. It allows us to study a model in which each pair of nodes can communicate, and we do not consider any details of how this communication is executed by the underlying network.

The \textbf{degree+1 list coloring (\DILC) problem} is as follows: the input is a graph $G = (V,E)$ on $n$ nodes, with color palettes $\Psi(u)$ assigned to each node $u \in V$, such that $\size{\Psi(u)} \geq d(u)+1$. The objective is to find a proper coloring of nodes in $G$ such that each node as assigned to a color from its color palette (where proper coloring means that no edge in $G$ is monochromatic). In \CONGESTEDC, each node $v$ of the input graph $G$ has assigned a network node and this network node knows $\Psi(v)$ and all neighbors of $v$ in $G$.

A useful property of the \CONGESTEDC model is that thanks to the constant-round routing algorithm of Lenzen \cite{Lenzen13}, information can be redistributed essentially arbitrarily in the communication network, so there is no need to associate the computational entities with nodes in the input graph $G$. (This is in stark contrast to the related LOCAL and CONGEST distributed models in which the link between computation and input graph locality is integral.) In particular, this allows us to collect graphs of size $O(n)$ on a single node in $O(1)$ rounds. Because of this ``decoupling'' of the computation from the input graph, where appropriate we will distinguish the nodes in their roles as computational entities (``network nodes'') from the nodes in the input graph (``graph nodes'').

Computation in \CONGESTEDC is not generally restricted, but in any case our algorithms use only polynomial computation at each node ($poly(n)$ assuming that the input size is $poly(n)$, i.e. that colors are named using at most $poly(n)$-bit labels; see \Cref{subsec:colorspace} below for discussion of color space).

\subsection{Notation}
\label{subsec:notations}

For $k \in \mathbb{N}$ we let $[k] := \{1,\dots,k\}$. We consider a graph $G = (V(G),E(G))$ with $V(G)$ as the node set and $E(G)$ as the edge set. The size of a graph $G$ refers to the number of edges in $G$ and is denoted by $|G|$. The set of neighbors of a node $v$ is denoted by $N_G(v)$ and the degree of a node $v$ is denoted by $d_G(v)$. The maximum degree in the graph $G$ is denoted by $\Delta_G$. For any node $v$, we partition its neighbors into two sets $N_G^+(v) := \{u \in N_G(v):d_G(u) > d_G(v)\lor (d_G(u) = d_G(v) \land ID(u)>ID(v))\}$ and $N_G^-(v) := \{u \in N_G(v):d_G(u) < d_G(v)\lor (d_G(u) = d_G(v) \land ID(u)<ID(v))\}$ (i.e. higher or lower-degree neighbors, where same-degree neighbors are partitioned based on ID). Let $d_G^+(v) := |N_G^+(v)|$ and $d_G^-(v) := |N_G^-(v)|$. When 
$G$ is clear from the context, we suppress $G$ from the subscripts of the notation. As already mentioned, our (recursive) algorithm \textsc{Color} uses three subroutines: \textsc{ColorTrial}, \textsc{SubSample}, and \textsc{BucketColor}. We use $G_0, G_1$, and $G_2$ to denote the subgraph of $G$ on which \textsc{ColorTrial}, \textsc{SubSample}, and \textsc{BucketColor} are executed, respectively. For a node $v$ in $G_0$, $N^{\gtrsim}(v) \subseteq N_{G_0}(v)$ denotes the subset of neighbors $u$ of $v$ in $G_0$ such that $d_{G_0}(u) \ge 3 d_{G_0}(v)$. Additionally, for a node $v$ in $G_1$, $N^{\approx} (v) \subseteq N_{G_1}(v)$ denotes the subset of  neighbors $u$ of $v$ in $G_1$  such that  $\frac{1}{2}d_{G_1}(v)\leq d_{G_1}(u) \leq 6d_{G_1}(v)$.

For the coloring problem, for a node $v$, $\Psi_G(v) \subseteq [n^{O(1)}]$ denotes the list of colors in the color palette of $v$ and $p_G(v) := |\Psi_G(v)|$. As we proceed in coloring the nodes of the input graph $G$ the graph will be changing and the color palettes of the nodes may also change. 

For binary strings $a$ and $a'$ in $\{0,1\}^*$, $a \sqsubseteq a'$ denotes that $a$ is a prefix of $a'$, and $a \sqsubset a'$ denotes that it is a \emph{strict} prefix of $a'$. Furthermore, $a' \sqsupseteq a$ iff $a \sqsubseteq a'$, and $a' \sqsupset a$ iff $a \sqsubset a'$. 

\subsection{Derandomization in \CONGESTEDC}
\label{subsec:mpc_derandomization}

The \emph{method of conditional expectations} using \emph{bounded-independence hash functions} is  nowadays a classical technique for the derandomization of algorithms \cite{ES73,Luby93,MNN94,Raghavan88}. Starting with the recent work of Censor-Hillel \etal \cite{CPS20}, this approach  has also been found  powerful in the setting of distributed and parallel algorithms, see e.g., \cite{BKM20,CDP21b,CDP21c,CDP21a,CDP20,DKM19,FGG22,GK18,Parter18}.

This technique requires that we show that our randomized algorithm can be made to work in expectation using only bounded-independence.
It is known that small families of bounded-independence hash functions exist, and that hash functions in these families can be specified by a short seed. It is also known that such a family must contain a hash function that is at least as good as the expectation due to the probabilistic method.
Using these facts, we can perform an efficient search for a hash function which is at least as good as the expectation by iteratively setting a larger and larger prefix of the seed of the hash function.

In this section we first give some useful lemmas regarding $O(1)$-wise independence and the existence of small families of $O(1)$-wise independent hash functions, and then we give a formal description of the method of conditional expectations and how it is implemented in the \CONGESTEDC model.

\subsection{Bounded Independence}
\label{subsec:bounded-independence}

Our algorithm will be finding a hash function of sufficient quality from a family of $O(1)$-independent hash functions. In the following, we recall the standard notions of \emph{$k$-wise independent hash functions} and \emph{$k$-wise independent random variables}.
Then, we recall that we can construct small families of bounded-independence hash functions, and that each hash function in this family can be specified by a short seed.
Finally, we state a concentration inequality (for $k$-wise independent random variables) that will be used in the analysis of our algorithm while bounding the probabilities of some desired events.


\begin{definition} \label{def:independence}
Let $k\geq 2$ be an integer. A set $\{X_1,\ldots,X_n\}$ of $n$ random variables taking values in $S$ are said to be \textbf{$k$-wise independent} if for any $I \subset [n]$ with $|{I}| \leq k$ and any $x_i \in S$ for $i \in I$, we have 
\begin{align*}
    \Prob{\bigwedge_{i \in I} X_i=x_i} &= \prod_{i\in I} \Prob{X_i=x_i}
    \enspace.
\end{align*}
\end{definition}

\begin{definition}
A family of hash functions $\mathcal{H} = \{h : X \rightarrow Y\}$ is said to be \textbf{$k$-wise independent} if $\{h(x): x \in X\}$ are $k$-wise independent when $h$ is drawn uniformly at random from $\mathcal{H}$.
\end{definition}

We use the property that small families of $O(1)$-wise independent hash functions can be constructed, and each hash function in such a family can be specified with a small number of bits:

\begin{lemma}[\cite{LubyRackoff1988,AlonMatiasSzegedy1999}]
\label{lem:families_of_hash_functions}
For all positive integers $c_1, c_2, k$, there is a family of $k$-wise independent hash functions $\mathcal{H} = \{h : [n^{c_1}] \rightarrow [n^{c_2}]\}$ such that each function from $\mathcal{H}$ can be specified using $O(k \log n)$ bits. These functions can be evaluated in $poly(k, \log n)$ computation.
\end{lemma}

The domain of these hash functions will be applied to represent node and color IDs, and the range $[n^{O(1)}]$ will be sufficient to represent the random choices we use in our algorithm (choice of color from nodes' palettes, membership of a subsampled set, choice of `bucket'), since each such choice can be made by assigning an appropriately-sized subset of $[n^{O(1)}]$ to each outcome, with negligible ($\le n^{-O(1)}$) rounding error. For example, we can use a hash function with range $[n^5]$ to choose colors u.a.r from a palette of size $P$ by partitioning the range into $P$ sets of $\lceil \frac{n^5}{P}\rceil$ or $\lfloor \frac{n^5}{P}\rfloor$ elements, with each set corresponding to one color choice. The marginal probability of each color choice is then $\frac{1}{P}\pm O(n^{-5})$, and so with high probability these $O(n^{-5})$ rounding errors in the probabilities do not affect any outcome.

We finish this subsection by giving some useful tail bounds for $O(1)$-wise independent random variables.

\begin{lemma}[Lemma 2.3 of \cite{BR94}]\label{lem:conc}
Let $k\ge 4$ be an even integer. Suppose $X_1,\dots, X_n$ are $k$-wise independent random variables taking values in $[0, 1]$. Let $X = X_1 + \dots + X_n$, and $\mu = \Exp{X}$ be the expectation of $X$. Then for any $A >0$, 
\begin{align*}
    \Prob{|X-\mu| \ge A} &\le 8\left(\frac{k\mu+k^2}{A^2}\right)^{k/2}
    \enspace.
\end{align*}
\end{lemma}

We will often use \Cref{lem:conc} with $k=100$ and $\mu \ge 1000$, giving the following special case:

\begin{corollary}\label{cor:conc}
Suppose $X_1,\dots, X_n$ are $100$-wise independent random variables taking values in $[0, 1]$. Let $X = X_1 + \dots + X_n$, and $\mu = \Exp{X}$ is the expectation of $X$. If $\mu \geq 1000$, then for any $A>0$:
\begin{align*}
    \Prob{|X-\mu| \ge A} &\le \left(\frac{111\mu}{A^2}\right)^{50}
    \enspace.
\end{align*}
\end{corollary}

\begin{proof}
By \Cref{lem:conc},
\begin{align*}
    \Prob{|X-\mu| \ge A} &\le 8\left(\frac{100\mu+10000}{A^2}\right)^{50} \le
    8 \left(\frac{110\mu}{A^2}\right)^{50} \le
    \left(\frac{111\mu}{A^2}\right)^{50}
    \enspace.
    \qedhere
\end{align*}
\end{proof}

\subsection{The Method of Conditional Expectations}
\label{subsec:method_of_conditional_expectations}

We now describe in more detail the method of conditional expectations and its implementation in \CONGESTEDC. We briefly recall the setup to the problem: we have a randomized algorithm which ``succeeds'' if a ``bad'' outcome occurs for less than some number $T$ of nodes. This algorithm succeeds in expectation using bounded-independence randomness. We would like to derandomize this algorithm. In order to achieve this, given a family of $O(1)$-wise independent hash functions $\mathcal{H}$, we need to find a ``good'' hash function $h^* \in \mathcal{H}$.

First, we define some cost function $f : \mathcal{H} \times V \rightarrow \{0,1\}$ such that $f(h, v)=1$ if the node $v$ has a ``bad'' outcome when $h$ is the selected hash function, and $f(h, v) = 0$ if the outcome is ``good''. We further define $F(h) = \sum_{v \in V} f(h, v)$ as the total cost of the hash function $h$, i.e., the number of bad nodes when $h$ is the selected hash function. Finally, we use $\mathbf{E}_{h \in \mathcal{H}}[F(h)]$ to denote the expected value of $F(h)$ when $h$ is drawn uniformly at random from~$\mathcal{H}$.

To successfully derandomize our algorithm, we need to find a hash function $h^* \in \mathcal{H}$ such that $F(h^*) \leq T$.
We need the following conditions to hold for our derandomization to work:
\begin{itemize}
\item $\mathbf{E}_{h \in \mathcal{H}} [F(h)] \leq T$ (i.e., the expected cost of a hash function selected uniformly at random from $\mathcal{H}$ is at most $T$); and
\item Node $v$ can locally (i.e., without communication) evaluate $f(h, v)$ for all $h \in \mathcal{H}$.
\end{itemize}

We can now use the method of conditional expectations to find a $h^* \in \mathcal{H}$ for which $F(h^*) \leq T$.
We first recall that each hash function in our family of $O(1)$-wise independent hash functions $\mathcal{H}$ can be specified using $O(\log n)$ bits, by \Cref{lem:families_of_hash_functions}.
Next, let $\Pi = \{0,1\}^{\log n}$ be the set of binary strings of length $\log n$, and for each $\pi \in \Pi$, let $\mathcal{H}_\pi$ denote the hash functions in $\mathcal{H}$ whose seeds begin with the prefix $\pi$.

Our goal is to find some seed-prefix $\pi \in \Pi$ for which $\mathbf{E}_{h \in \mathcal{H}_\pi} [F(h)] \leq T$: the existence of such a prefix is guaranteed by the probabilistic method.
Since each node $v$ can locally evaluate $f(h, v)$ for all $h \in \mathcal{H}$, nodes can also compute $\mathbf{E}_{h \in \mathcal{H}_\pi} [f(h, v)]$  for all $\pi \in \Pi$. Since $|\Pi| = n$, each node $v$ can be made responsible for a prefix $\pi_v \in \Pi$. Node $v$ can then collect the value of $\mathbf{E}_{h \in \mathcal{H}_{\pi_v}} [f(h, u)]$ for each $u \in V \setminus \{v\}$: since this requires all nodes sending and receiving $O(n)$ messages it can be done in $O(1)$ rounds using Lenzen's routing algorithm \cite{Lenzen13}.
Now, by linearity of expectation:
$$
    \sum_{u \in V} \left(\mathbf{E}_{h \in \mathcal{H}_{\pi_v}} \left[f(h, u)\right]\right) = \mathbf{E}_{h \in \mathcal{H}_{\pi_v}} \left[F(h)\right]
    \enspace.
$$

Therefore $v$ can compute the expected value of $F$ for the sub-family of hash functions which are prefixed with $\pi_v$. Nodes can broadcast this expected value to all other nodes in $O(1)$ rounds, again using Lenzen's routing algorithm \cite{Lenzen13}. All nodes then know the expected value of $F$ for all $\log n$-bit prefixes and can, without communication (breaking ties in a predetermined and arbitrary way), pick the prefix with the lowest expected value of $F$. Recall that this prefix is guaranteed to have an expected value of at most $T$ by the probabilistic method.

We have now fixed the first $\log n$ bits of the prefix and obtained a smaller set $\mathcal{H}_1 \subset \mathcal{H}$ of hash functions. The same procedure (discussed above for $\mathcal{H}$) can then be applied to $\mathcal{H}_1$ to set the next $\log n$ bits of the seed, obtaining a smaller set $\mathcal{H}_2 \subset \mathcal{H}_1 \subset \mathcal{H}$ of hash functions. After repeating this procedure $O(1)$ times we will have fixed the entire seed, since we fix $\log n$ bits each time and the seeds of hash functions in $\mathcal{H}$ are $O(\log n)$ bits in length.

In our algorithm, the local cost functions $f(h,v)$ will be weighted indicator variables that node $v$ is in some way `bad' or `failed' (which will have a specific technical definition in each instance in which we use it). Then, the global cost function $F(h)$ will provide a size bound on the induced graph of bad or failed nodes. In this way, our general argument will be that we can deterministically fix hash functions that cause almost all nodes to satisfy some good properties, and the graph of nodes that do not satisfy those properties is small enough to collect to a single network node and deal with later.

\subsection{A Note on Color Space}
\label{subsec:colorspace}

We assume in the remainder of this work that the colors available to nodes are drawn from $[n^{O(1)}]$, which ensures that they can be specified within the $O(\log n)$ bits of a message. In fact, if the colors are instead drawn from $[2^{n^{O(1)}}]$, we can perform an initial $O(1)$-round deterministic color-renaming process to reduce the color space to $[O(n^{4})]$; we sketch this process as follows.

We use an $\epsilon$-almost pairwise independent family of hash functions $[2^{n^{O(1)}}] \rightarrow [O(n^{4})]$, with $\epsilon = n^{-3}$. This is a family of hash functions such that the outputs of a function chosen uniformly at random satisfy $2$-wise independence as in Definition \ref{def:independence} with additive error $\pm \epsilon$. It is known that there exists such a family of size $n^{O(1)}$ (see e.g. \cite{KJS97}, as used in \cite{CDP21b}). Under this random hash function, the probability that there exists a node with two colors in its palette hashed to the same value is at most $O(n^{-2})$. The method of conditional expectations, as above, can then be used to deterministically select a hash function such that each node's palette colors indeed map to $[O(n^{4})]$ with no collisions. 

Now, we can proceed to perform \DILC under this color space using our algorithm, and nodes can then map back to the original color space for their final output. Note that there may have been multiple colors \emph{across the entire input graph} that mapped to the same hashed value, but since we ensured that no two of these colors were present in the same node's palette, this will not cause any coloring conflicts.

It remains open whether efficient \DILC can be performed when colors are drawn from a larger space than $[2^{n^{O(1)}}]$.

\section{The \DILC Algorithm}
\label{sec:D1LC-algorithm}

The framework of our \CONGESTEDC algorithm is \textsc{Color$(G,x)$} (\Cref{alg:mainrand}), which colors graph $G$ relying on three main procedures: \textsc{ColorTrial}, \textsc{SubSample}, and \textsc{BucketColor}.

\textsc{ColorTrial} is a derandomized version of a simple and frequently used coloring procedure: all nodes nominate themselves with some constant probability, and nominated nodes then pick a color from their palette. If no neighbors choose this same color, the node is successful and takes this color permanently. For our algorithm, the goal of \textsc{ColorTrial} is to provide permanent slack for nodes $v$ whose neighbors mostly have significantly higher degree than their own. Since these neighbors have significantly higher degrees than $v$, they either have a larger palette than $v$ (in which case they likely color themselves using a color not in $v$'s palette, providing slack to $v$), or $v$ has a significantly larger palette size than degree (in which case it already has slack). 

\textsc{SubSample} is a derandomized version of sampling: nodes $v$ \emph{defer} themselves to $S$ (to be colored later) with probability $d(v)^{-0.1}$. The purpose of this is to provide temporary slack to nodes whose neighbors mostly have similar degrees to their own. We will then recursively run the whole algorithm on $S$, and we will show that after $O(1)$ recursive calls the remaining graph will be of size $O(n)$, which can be trivially colored in \CONGESTEDC in $O(1)$ rounds.

\textsc{BucketColor} is our main coloring procedure, and is designed to color all nodes for which \textsc{ColorTrial} and \textsc{SubSample} have generated sufficient slack, as well as all nodes whose neighbors mostly have lower degree than their own.

All these three algorithms begin with a randomized procedure, and use the method of conditional expectations on a family of $O(1)$-wise independent hash functions to derandomize it. Note that \emph{this derandomization is an essential part of the algorithm} even if one is only concerned with probabilistic success guarantees. This is because in low-degree graphs, we cannot obtain the necessary properties
with high probability, and some nodes will fail. The method of conditional expectations ensures that these graphs of failed nodes are of $O(n)$ size (and hence can be collected onto a single network node in $O(1)$ rounds to color sequentially).

\hide{
Using \textsc{ColorTrial}, \textsc{SubSample}, and \textsc{BucketColor}, we  present our main Algorithm \textsc{Color$(G,x)$}, in Algorithm~\ref{alg:mainrand}, ???\Anote{A nice formulation of what \textsc{Color$(G,x)$} does wrt $x$.} \Anote{We will want to define the meaning of $x$ and $C$ (some sufficiently large constant) in \textsc{Color$(G,x)$}.}
\Anote{Didn't we want to start \textsc{Color$(G,x)$} with reducing the max-degree to $O(\sqrt{n})$?} that colors a graph $G$ with $\Delta_G \leq O(\sqrt{n})$.~\footnote{Only \textsc{ColorTrial} requires $\Delta_G \leq O(\sqrt{n})$ to be executed in $O(1)$ rounds.} Note that, here parameter $x$ is an input to \textsc{Color$(G,x)$} such that $0\leq x\leq 0.9$ and it quantifies the size of the remaining graph over iteration. Without loss of generality let us assume that the size of $G$ is $\omega(n)$, as otherwise the entire graph can be collected onto a single network node in $O(1)$ rounds and the coloring can be done locally. Note that \textsc{Color$(G,x)$} first removes any node in whose degree is at most $C$ in $G$, which will be helpful to show the properties we require for \textsc{ColorTrial} and \textsc{SubSample}. Then \textsc{Color$(G,x)$} executes \textsc{ColorTrial} \textsc{SubSample} and \textsc{BucketColor} in sequence as described in Algorithm~\ref{alg:mainrand}. Then it recursively calls \textsc{Color$(G',x+0.1)$}. At the end it colors the nodes that are deferred either due to being degree at most $C$ in $G$ or being failed in \textsc{ColorTrial}. It is important to note that, the induced subgraph by all the deferred nodes is $O(n)$: hence can be gathered onto single network node in $O(1)$ rounds and can be colored locally.
}

Using \textsc{ColorTrial}, \textsc{SubSample}, and \textsc{BucketColor}, we can present our main algorithm \textsc{Color$(G,x)$} (\Cref{alg:mainrand}) to color graph $G$. The algorithm assumes that $\Delta_G \leq O(\sqrt{n})$
and it uses a parameter $x$, $0\leq x\leq 0.9$, that quantifies the size of the remaining graph over recursive calls (the algorithm starts with $x=0$ and recursively increases by $0.1$ until $x = 1$). In \Cref{sec:pf-main} (\Cref{lemma:reduce-to-sqrt}), we extend the analysis to arbitrary graphs, allowing arbitrary $\Delta_G$.

\begin{algorithm}[H]
\caption{\textsc{Color$(G,x)$}: $\Delta_G \leq O(\sqrt{n})$; $0\leq x \leq 0.9$; $C$ is a sufficiently large constant}
\label{alg:mainrand}

If $|G| = O(n)$, then collect $G$ in a single network node and solve the problem locally.
\label{step1-alg:mainrand}

Set $L_0 := \{v \in G: p_G(v)<C\}$ and $G_0 := G \setminus L_0$.
\label{step2-alg:mainrand}

$G_1, F \gets \textsc{ColorTrial$(G_0)$}$.  
\label{step3-alg:mainrand}
	
$G',G_2 \gets \textsc{Subsample$(G_1,x)$}$.  
\label{step4-alg:mainrand}
	
$\textsc{BucketColor$(G_2)$}$.
\label{step5-alg:mainrand} \textit{\small \hfill // It properly colors all nodes in $G_2$.}
	
$\textsc{Color$(G',x+0.1)$}$.  
\label{step6-alg:mainrand}
	
Collect and solve first $L_0$ and then $F$ at a single node.
\label{step7-alg:mainrand}
\end{algorithm}

Step \ref{step1-alg:mainrand} in \textsc{Color$(G,x)$} uses the fact that if $G$ is of size $O(n)$, then in \CONGESTEDC, the entire graph can be collected onto a single network node in $O(1)$ rounds and the coloring can be done locally. In the same way, since $L_0$ consists of vertices of constant degree, we can color them in step \ref{step7-alg:mainrand} in $O(1)$ rounds. Similarly, we will argue that the graph $F$ (of failed nodes in \textsc{ColorTrial}) is of size $O(n)$, and hence it can be colored in step \ref{step7-alg:mainrand} in $O(1)$ rounds. The central part of our analysis will be to show that after a constant number of recursive calls the algorithm terminates with a correct solution to \DILC of $G$.

To prove the correctness of our algorithm, we show the following properties of \textsc{Color}$(G,x)$:
\begin{enumerate}
\item \textsc{ColorTrial}, \textsc{SubSample}, and \textsc{BucketColor} run deterministically in $O(1)$ rounds.
\item The size of $F$ is $O(n)$.
\item Each node in $G_2$ has sufficient slack to be colored by \textsc{BucketColor}. For each node $v$ of $G_2$,  either $p_{G_2}(v) \ge  d_{G_2}(v)+\frac 14 d_{G_2}(v)^{0.9} $, or $|N_{G_2}^-(v)|\ge \frac 13d_{G_2}(v)$.
\item The size of the (remaining) graph reduces over recursive calls in the following sense:
    \begin{align}
    \label{ineq:property-of-G'}
        \sum_{v\in G'} d_{G_1}(v)^{x+0.1}
            &\le
        C n + 2 \sum_{v\in G_1} d_{G_1}(v)^{x}
        \enspace.
    \end{align}
Note that when $x = 0.9$, the fact that $G'$ is a subgraph of $G_1$ implies that the number of edges in $G'$ is bounded by expression (\ref{ineq:property-of-G'}).
 In particular, we show that the total size of the remaining graph is $O(n)$ after $10$ recursive calls.
\end{enumerate}

In \Cref{sec:col-samp}, we describe the procedures \textsc{ColorTrial} and \textsc{SubSample}. Also, we prove the desired properties of $F$, $G_2$, and $G'$ in \Cref{sec:col-samp}. In \Cref{sec:buck}, we present the procedure \textsc{BucketColor}, which properly colors all nodes in $G_2$. 
 Finally, we prove our main theorem (\Cref{thm:D1LCcolor}) in \Cref{sec:pf-main}.

For simplicity of the presentation, in the pseudocode of our algorithms in the following sections, we will only present the randomized bases of each procedure. In each case, the full deterministic procedure comes from applying the method of conditional expectations to the randomized bases, with some specific cost function we will make clear in the analysis.

\section{\textsc{ColorTrial} and \textsc{SubSample}}
\label{sec:col-samp}

We  describe procedure \textsc{ColorTrial} and \textsc{SubSample} in \Cref{sec:trial} and \Cref{sec:samp}, respectively, along with some of their crucial useful properties. In particular, we show that $|F|=O(n)$ in \Cref{lem:sizeF} of \Cref{sec:trial} and show that graph $G'$ has the desired property  in \Cref{lem:G'size} of \Cref{sec:samp}. Finally, we give a lemma capturing the desired property of graph $G_2$ (which is the input to \textsc{BucketColor}), and we prove it in \Cref{ssec:G-2}.
\subsection{\textsc{ColorTrial}}
\label{sec:trial}

We first note that, because nodes with palette size less than $C$ are removed immediately prior to \textsc{ColorTrial}($G_0$) in \textsc{Color}($G,x$), we may assume that all nodes $v$ in $G_0$ have $p_{G_0}(v)\ge C$. The randomized procedure on which \textsc{ColorTrial} is based is Algorithm~\ref{alg:colortrialrand}.  To derandomize \textsc{ColorTrial}, we replace each of the random choices of lines 1 and 3 (which we will call the nomination step and the coloring step respectively) with choices determined by a hash function from a $100$-wise independent family $[n^{O(1)}]\rightarrow [n^{O(1)}]$. \textsc{ColorTrial} has two major steps: nomination step (line 1) and coloring step (line 3). The coloring of a node is deferred if it \emph{fails} either in the nomination step or in the coloring step. Note that, the notions of failure are different in the nomination step and the coloring step, and we will define  the notions in \Cref{def:fail-nom} and \Cref{def:fail-col} respectively.

\begin{algorithm}[H]
\caption{\textsc{ColorTrial$(G_0)$} - Randomized Basis}
\label{alg:colortrialrand}
Each node $v$ in $G_0$ independently self-nominates with probability $\frac 14$.\label{step:nom}

Each node $v$ decides if it is successful or failed in the nomination step according to \Cref{def:fail-nom}.


For each node $v$ (that is successful in the nomination step):
{\begin{itemize}
\item $v$ chooses a random palette color $c(v)\in \Psi_{G_0}(v)$; \label{step:col}
\item $v$ colors itself with color $c(v)$ if no neighbor $u$ of $v$ chose $c(u) = c(v)$;
\item $v$ decides if it is successful or failed in the coloring step according to \Cref{def:fail-col}.  \hfill 
\end{itemize}}
	
Return
{\begin{itemize}
\item $G_1$, the induced graph of remaining (non-failed) uncolored nodes, with updated palettes,
\item $F$, the induced graph of failed nodes (either in the nomination step or in the coloring step), with updated palettes.
\end{itemize}}
\end{algorithm}

First, we define some notions (in \Cref{def:inter}) that will be useful to define failed nodes in the nomination step and the coloring step of \textsc{ColorTrial}.

\begin{definition}\label{def:inter}
Recall that $N^{\gtrsim}(v) \subseteq N_{G_0}(v)$ is defined as the subset of neighbors $u$ of $v$ that have $d_{G_0}(u) \ge 3 d_{G_0}(v)$.  $\mathrm{Nom}_v\subseteq N_{G_0}(v)$ is defined as the subset of $v$'s neighbors that self-nominate, and $\mathrm{Nom}^{\gtrsim}_v  := \mathrm{Nom}_v\cap N^{\gtrsim}(v)$.
\end{definition}

Next, we define the notion of failed nodes in the nomination step.
%
\begin{definition} \label{def:fail-nom}
A node $v$  in $G_0$ is \textbf{successful} during the nomination step of \textsc{ColorTrial} if both of the following hold:
\begin{itemize}
\item Few nominated neighbors: $|\mathrm{Nom}_v| \le \frac 14 d_{G_0}(v) + p_{G_0}(v)^{0.7}$;
\item Sufficiently many nominated neighbors  with higher degrees: $|\mathrm{Nom}^{\gtrsim}_v| \ge \frac 14|N^{\gtrsim}(v)| - p_{G_0}(v)^{0.7}$.
\end{itemize}
Node $v$ \textbf{fails} in the nomination step if neither of the above conditions holds.
\end{definition}

In \Cref{lem:nominations}, we show that if we choose a hash function uniformly at random from a $100$-wise independent hash family, the subgraph induced by the failed nodes in the nomination step is of size $O(n)$ in expectation. Then, in \Cref{lem:MCE}, we derandomize this selection of a hash function using the method of conditional expectations.
\begin{lemma}
\label{lem:nominations}
When nomination choices of \textsc{ColorTrial} are determined by a random hash function from a $100$-wise independent hash family $[n^{O(1)}]\rightarrow [n^{O(1)}]$, the probability that any fixed node $v$ in $G_0$ fails in the nomination step of \textsc{ColorTrial} is at most $1/p_{G_0}(v)$. 
\end{lemma}

\begin{proof}
Recall \Cref{def:fail-nom} and note that a node $v$ can fail in the nomination step due to two reasons. We consider the two cases of a node being failed separately.

\paragraph{Event 1: $|\mathrm{Nom}_v| > \frac 14 d_{G_0}(v) + p_{G_0}(v)^{0.7}$.}
For $u \in N_{G_0}(v)$, let $X_u$ denote the indicator random variable for the event that $u$ self-nominates. So, $\mathrm{Nom}_v=\sum_{u \in N_{G_0}(v)}X_u$.

Note that $\Pr [X_u=1]=1/4$ and the expected value of $\size{\mathrm{Nom}_v}$ is $\mu=\frac14 d_{G_0}(v)$. We may assume $\mu$ is at least $1000$, since $p_{G_0}(v)\ge C$ for sufficiently large $C$, and therefore if $d_{G_0}(v)$ is not at least a sufficiently large constant, we already have $|\mathrm{Nom}_v|\le d_{G_0}(v) \le \frac14 d_{G_0}(v) + p_{G_0}(v)^{0.7}$. Also, observe that the set of random variables $\{X_u: u \in N_{G_0}(v)\}$ are $100$-wise independent when the randomness for the procedure is provided by a random hash function from a $100$-wise independent hash family.

So, applying \Cref{cor:conc} (with $X=|\mathrm{Nom}_v|$, $\mu=\frac14 d_{G_0}(v)$, and $A= p_{G_0}(v)^{0.7}$), we deduce the following:
\begin{align*}
	\Prob{|\mathrm{Nom}_v| > \tfrac 14 d_{G_0}(v) + p_{G_0}(v)^{0.7}}
        &\le
    \left(\frac{28 d_{G_0}(v)}{p_{G_0}(v)^{1.4}}\right)^{50}
        \le 
    \left(28 p_{G_0}(v)^{-0.4}\right)^{50}
        < 
    \tfrac12  p_{G_0}(v)^{-1}
    \enspace.
\end{align*}

\paragraph{Event 2: $|\mathrm{Nom}^{\gtrsim}_v| < \frac 14|N^{\gtrsim}(v)| - p_{G_0}(v)^{0.7}$.}

The proof is similar to that of Event 1. Note that the expected value of $|\mathrm{Nom}^{\gtrsim}_v| $ is $\frac 14 N^{\gtrsim}(v)$. We may again assume $\mu$ is at least $1000$, since $p_{G_0}(v)\ge C$ for sufficiently large $C$, and therefore if $N^{\gtrsim}(v)$ is not sufficiently large constant, we trivially have $|\mathrm{Nom}^{\gtrsim}_v|\ge 0 > \frac14 N^{\gtrsim}(v) - p_{G_0}(v)^{0.7} $.

So, applying \Cref{cor:conc} (with $X=|\mathrm{Nom}^{\gtrsim}_v|$,  $\mu=\frac14 |N^{\gtrsim}(v)|$, and $A= p_{G_0}(v)^{0.7}$), we deduce the following:
\begin{align*}
    \Prob{|\mathrm{Nom}^{\gtrsim}_v| < \tfrac14 |N^{\gtrsim}(v)| - p_{G_0}(v)^{0.7}}
        &\le
    \left(\frac{28 |N^{\gtrsim}(v)|}{p_{G_0}(v)^{1.4}}\right)^{50}
        \le
    \left(28 p_{G_0}(v)^{-0.4}\right)^{50}
        <
    \tfrac12  p_{G_0}(v)^{-1}
    \enspace.
\end{align*}
So, the total probability that $v$ fails in the nomination step  is at most $\frac12  p_{G_0}(v)^{-1}+\frac12  p_{G_0}(v)^{-1}= ({p_{G_0}(v)})^{-1}$.
\end{proof}

\begin{lemma}\label{lem:MCE}
We can deterministically choose a hash function in $O(1)$ rounds, from a $100$-wise independent family $[n^{O(1)}]\rightarrow [n^{O(1)}]$, to run the nomination step of \textsc{ColorTrial} such that the size of the subgraph induced by the failed nodes (in the nomination step) is~$O(n)$.
\end{lemma}
\begin{proof}
For each vertex $v \in G_0$, define a random variable
\[
I_v =
\begin{cases}
d_{G_0}(v), & \text{if $v$ is a failed node in the nomination step},\\
0, & \text{otherwise.}
\end{cases}
\]

Our local cost function $f(h,v)$ for the method of conditional expectations takes the value of $I_v$ when coloring choices are specified by $h$, and the global cost function $F(h) = \sum_{v \in G_0}f(h,v)$. Observe that the size of the required subgraph (induced by the failed nodes in the nomination step) is bounded by $F(h) = \sum_{v \in G_0}I_v$. By \Cref{lem:nominations} and linearity of expectation,  the expected size of the required subgraph is at most $\sum_{\text{$v\in G_0$}}d_{G_0}(v) \cdot 1/p_{G_0}(v) \le n$. Note that $\sum_{v \in G_0}I_v$ is  the aggregate of $I_v$'s. Each node $v$ can compute $I_v$ by checking the conditions mentioned in \Cref{def:fail-nom} if the $1$-hop neighborhood of a node $v$ is known, which is the case in \CONGESTEDC. The method of conditional expectations applied to $F(h)$ (outlined in \Cref{subsec:method_of_conditional_expectations}) implies that we can find a hash function deterministically so that the size of the subgraph induced by the failed nodes (in the nomination step) is at most $n$.
\end{proof}

	Besides the nomination step, a node can also fail in the coloring step  of \textsc{ColorTrial}. Now we formally define what it means for a node to fail in the coloring step.
	
	\begin{definition} \label{def:fail-col}
		A node $v$ in $G_0$ is \textbf{successful} during the coloring step of \textsc{ColorTrial} if at least one of the following hold:
		\begin{enumerate}
			\item $p_{G_0}(v) \ge 1.1 d_{G_0}(v)$;
			\item $\size{N^{\gtrsim}(v)} < \frac 13 d_{G_0}(v)$;
			\item at least $0.03d_{G_0}(v)$ of $v$'s neighbors failed in the nomination step; 
			\item at least $0.01p_{G_0}(v)$ of $v$'s neighbors successfully color themselves a color not in $v$'s palette.
		\end{enumerate}
		Node $v$ is classed as \textbf{failed} in the coloring step if none of the above four properties hold.
	\end{definition}
	
	Notice that the first three properties are already determined by the nomination step. Here, we handle the fourth property in our analysis. Similar to our analysis for the nomination step, in \Cref{lem:colorstep}, we show that choosing a hash function uniformly at random from a $100$-wise independent hash family to make decisions in the coloring step yields a subgraph of failed nodes of size $O(n)$ in expectation. Then we show that one can derandomize this selection of a hash function using the method of conditional expectations, achieving \Cref{lem:MCE2}.

Consider a node $v$ that satisfies none of the first three properties of \Cref{def:fail-col}, i.e., $p_{G_0}(v) < 1.1 d_{G_0}(v)$, $|N^{\gtrsim}(v)| \ge \frac 13 d_{G_0}(v)$, and fewer than $0.03d_{G_0}(v)$ of $v$'s neighbors failed in the nomination step.

Let $P_v$ be the set of nodes in $\mathrm{Nom}^{\gtrsim}_v$ that succeeded in the nomination step. 	

\begin{claim}\label{cl:cru}
	$\size{P_v}> 0.04 d_{G_0}(v)$.
\end{claim}
\begin{proof}
	Firstly, we have $|\mathrm{Nom}^{\gtrsim}_v|\ge \frac 14|N^{\gtrsim}(v)| - p_{G_0}(v)^{0.7}$ since $v$ succeeded in the nomination step. Therefore, by the assumption that it did not satisfy properties 1-3 of \Cref{def:fail-col},
\begin{align*}
	|\mathrm{Nom}^{\gtrsim}_v| &\ge
	\tfrac{1}{12}d_{G_0}(v) - p_{G_0}(v)^{0.7} >
	\tfrac{1}{12\cdot 1.1} p_{G_0}(v) - p_{G_0}(v)^{0.7}
	\enspace.
\end{align*}
Since we may further assume that $p_{G_0}(v)\ge C$ for sufficiently large constant $C$ (or $v$ would have been moved to $L_0$),
\begin{align*}
	|\mathrm{Nom}^{\gtrsim}_v| &> 0.07 p_{G_0}(v)
	\enspace.
\end{align*}

As the number of $v$'s neighbors failed in the nomination step is at most $0.03d_{G_0}(v)$,\hide{\gopi{May be this is ok? There is a discrepancy between this value of $0.03$ and the value $0.014$ in the third point in  definition 4.5?} Thanks - the $0.03$ value is the correct one.} there are more than $0.04p_{G_0}(v)$ nodes in $\mathrm{Nom}^{\gtrsim}_v$ that succeeded in the nomination step, i.e., $|P_v|> 0.04p_{G_0}(v)$. \end{proof}
	
Consider an arbitrary subset $S$ of $d_{G_0}(v)^{0.4}$ nodes in $P_v$. Let $Q$ denote the set of nodes of $S$ that choose the same color as another node in $S$, and let $q:=|Q|$. We bound the size of $Q$ as follows:

\begin{claim}\label{claim:100}
When color choices of \textsc{ColorTrial} are determined by a random hash function from
a $100$-wise independent hash family $[n^{O(1)}] \rightarrow [n^{O(1)}]$, $\Exp{q^{50}}< 1$.
\end{claim}

\begin{proof}
	When $q$ nodes share a color with another node in $S$, then we must be able to find some acyclic set of at least $q/2$ monochromatic edges within $S$ (by taking a spanning tree of each connected group of monochromatic nodes). For values $q>100$, then from this set, consider taking a further subset of $50$ edges. There are $\binom{q/2}{50}\ge (\frac{q}{100})^{50}$ choices of such a subset, and therefore there are at least $(\frac{q}{100})^{50}$ different acyclic sets of $50$ monochromatic edges.
	
	We now bound the expected number of acyclic subsets of $50$ monochromatic edges within $S$. There are fewer than $\binom{|E(S)|}{50} < \binom{d_{G_0}(v)^{0.8}}{50} \le d_{G_0}(v)^{40}$ subsets of $50$ edges in $S$. If we fix some particular acyclic set of $50$ edges, the probability that all those edges are indeed monochromatic is at most $(3d_{G_0}(v))^{-50}$, since each $u\in P_v\subseteq \mathrm{Nom}^{\gtrsim}_v$ has palette size greater than $3d_{G_0}(v)$, and this probability holds under  $100$-wise independence since it depends only on the colors of at most $100$ nodes.
	
	So, the expected number of acyclic sets of $50$ monochromatic edges is at most $ d_{G_0}(v)^{40} \cdot (3d_{G_0}(v))^{-{50}}<d_{G_0}(v)^{-10}$. The probability that there are indeed the $(\frac{q}{100})^{50}$ such subsets required is at most $\frac{d_{G_0}(v)^{-10}}{(\frac{q}{100})^{50}}$ by Markov's inequality.
	
	For values $q\le 100$, we instead bound the probability that the set occurs directly. There are fewer than $\binom{|E(S)|}{q/2} < \binom{d_{G_0}(v)^{0.8}}{q/2} \le d_{G_0}(v)^{0.4q}$ acyclic subsets of $q/2$ edges in $S$, and the probability that each has entirely monochromatic edges is at most $(3d_{G_0}(v))^{-q/2}$. So, the expected number of acyclic sets of $q/2$ monochromatic edges is at most $ d_{G_0}(v)^{0.4q} \cdot (3d_{G_0}(v))^{-{q/2}}<d_{G_0}(v)^{-0.1q}$, and the probability that any such set exists is at most $d_{G_0}(v)^{-0.1q}$ by Markov's inequality.
	
	Then, 
	\begin{align*}
		\Exp{q^{50}} &= \sum_{x=2}^{100} x^{50}\Prob{q=x}+\sum_{x=101}^{|S|} x^{50}\Prob{q=x}\\
		&\le \sum_{x=2}^{100} x^{50}d_{G_0}(v)^{-0.1x}+\sum_{x=101}^{d_{G_0}(v)^{0.4}} x^{50}\frac{d_{G_0}(v)^{-10}}{(\frac{x}{100})^{50}}\\
		&\le d_{G_0}(v)^{-0.1}+100^{50}d_{G_0}(v)^{-9}\\
		&<\enspace 1.
	\end{align*}
	
\end{proof}

For each color $ c\in \Psi_{G_0}(u)$, define:

\[w_u(c) = \sum\limits_{w\in Nom_u\setminus S: c\in \Psi_{G_0}(w)}p_{G_0}(w)^{-1}\enspace,\] i.e. the expected number of (nominated) neighbors of $u$ outside of $S$ that will pick color $c$. We aim to minimize this quantity, and first show an upper bound on its expectation:

\begin{claim}\label{claim:avg}
When color choices of \textsc{ColorTrial} are determined by a random hash function from
a $100$-wise independent hash family $[n^{O(1)}] \rightarrow [n^{O(1)}]$, $\Exp{\sum_{u\in S} w_u(c(u))} \le 0.26d_{G_0}(v)^{0.4}$.
\end{claim}

\begin{proof}
For each $u\in S$, $\Exp{w_u(c(u))} = p_{G_0}(u)^{-1}\sum\limits_{c\in \Psi_{G_0}(u)} w_u(c)$. Each nominated neighbor $w\in \mathrm{Nom}_u$ contributes at most $1$ to the total sum, and since $w$ succeeded in the nomination stage, by Definition \ref{def:fail-nom} it has $|\mathrm{Nom}_u| \le \frac 14 d_{G_0}(u) + p_{G_0}(u)^{0.7}$. So, $\Exp{\sum_{u\in S} w_u(c(u))}\le p_{G_0}(u)^{-1} \cdot (\frac 14 d_{G_0}(u) + p_{G_0}(u)^{0.7})< \frac14 + p_{G_0}(u)^{-0.3} \le 0.26 $ (since $u\in N^{\gtrsim}(v)$, we may assume $p_{G_0}(u)\geq p_{G_0}(v) \geq C$ for sufficiently large constant $C$, as otherwise $v$ would have been moved to $L_0$). So,

$\Exp{\sum_{u\in S} w_u(c(u))}\le 0.26 |S| = 0.26d_{G_0}(v)^{0.4}$.
\end{proof}

We would like to show that the number of monochromatic edges from nodes in $S$ concentrates around its expectation. However, this is not quite the case: it is possible for some nodes to choose colors with very high $w_u(c(u))$ values, and then their number of adjacent monochromatic edges can be large. However, this is not a problem overall, because there will be few such nodes. So, we will first define a set $S'$ of nodes with low $w_u(c(u))$: let $S' :=\{u\in S:w_u(c(u))\le d_{G_0}(v)^{0.1}\}$. $S'$ depends on the random color choices, but we can show that it contains most nodes in $S$.

\begin{claim}\label{claim:S'}
When color choices of \textsc{ColorTrial} are determined by a random hash function from a $100$-wise independent hash family $[n^{O(1)}] \rightarrow [n^{O(1)}]$, with probability at least $1-d_{G_0}(v)^{-5}$, $|S'|\ge |S|-0.27d_{G_0}(v)^{0.35}$.

\end{claim}

\begin{proof}

	The probability that $u\in S'$, i.e. that $u$ chooses a color $c(u)$ with $w_u(c(u)) \ge d_{G_0}(v)^{0.1}$, is at most $0.26 d_{G_0}(v)^{-0.1}$ by Markov's inequality. Denote this event $H_u$. We bound the probability that $H_u$ occurs for more than $0.27d_{G_0}(v)^{0.35}$ nodes $u \in S$.
	
	Consider an arbitrary subset $A$ of $100$ nodes $u\in S$. There are $\binom{|S|}{100}$ such sets. The probability that $H_u$ occurs for all nodes in $A$ is at most $(0.26 d_{G_0}(v)^{-0.1})^{100}  $, even under only $100$-wise independence. So, the expected number of such subsets $A$ with $H_u$ occurring for all $u\in A$ is at most $\binom{|S|}{100}\cdot (0.26 d_{G_0}(v)^{-0.1})^{100}$. By Markov's inequality, with probability at least $1-d_{G_0}(v)^{-5}$, there are at most $\binom{|S|}{100}\cdot (0.26 d_{G_0}(v)^{-0.05})^{100}$ such sets $A$. In this instance, denoting $a=|\{u\in S:H_u\}|$, we have $\binom{a}{100}/\binom{|S|}{100} \le(0.26 d_{G_0}(v)^{-0.05})^{100}$. But,
	
	\[\frac{\binom{a}{100}}{\binom{|S|}{100}} = \frac{a!(|S|-100)!}{|S|!(a-100)! } = \frac{a}{|S|}\cdot \frac{a-1}{|S|-1}\cdot \dots \cdot  \frac{a-100}{|S|-100 }> \left(\frac{a-100}{|S|-100 }\right)^{100}\enspace.\]
	
	So, 
	
	\[a < 0.26 d_{G_0}(v)^{-0.05}(|S|-100 ) +100 \le 0.26 d{G_0}(v)^{0.35} +100 \le 0.27d_{G_0}(v)^{0.35},\]
	
	again using that $d_{G_0}(v)$ is at least sufficiently large constant $C$.

\end{proof}

Now, we show that the sum of $w_u(c(u))$ values in $S'$ concentrates around its expectation. To do so, consider the following variables: \[X_u = \begin{cases}
	w_u(c(u)) \cdot d_{G_0}(v)^{-0.1}  &\text{if $w_u(c(u))\le d_{G_0}(v)^{0.1}$ (i.e. if $u\in S'$)}\\
	0&\text{otherwise}
\end{cases}\enspace.\] Note that $X_u \in [0,1]$. Let $X:= \sum_{u\in S}X_u$, let $\mu_u := \Exp{X_u}$, and let $\mu_X:= \Exp{X}$.

\begin{claim}\label{claim:S'2}
When color choices of \textsc{ColorTrial} are determined by a random hash function from
a $100$-wise independent hash family $[n^{O(1)}] \rightarrow [n^{O(1)}]$, $\Exp{(X-\mu)^{100}}\le 10^{80} d_{G_0}(v)^{20}$
\end{claim}

\begin{proof}
By linearity of expectation,

\begin{align*}
\Exp{(X-\mu_X)^{100}}& = \Exp{\left(\sum_{u\in S} X_u - \mu_u \right)^{100}}\\
&= \Exp{\sum_{u_1, u_2, \dots, u_{100}\in S} (X_{u_1} - \mu_{u_1})  (X_{u_2} - \mu_{u_2}) \dots  (X_{u_{100}} - \mu_{u_{100}})}\\
&= \sum_{u_1, u_2, \dots, u_{100}\in S}\Exp{ (X_{u_1} - \mu_{u_1}) (X_{u_2} - \mu_{u_2}) \dots (X_{u_{100}} - \mu_{u_{100}})}\enspace.
\end{align*}

Notice that $\Exp{(X-\mu_X)^{100}}$ is the sum of expectations of terms involving up to $100$ nodes (some sets contain duplicates). Each of these sets of up to $100$ nodes behave independently under a hash function from a $100$-wise independent family. Therefore, the value of $\Exp{(X-\mu)^{100}}$ is the same as under \emph{fully independent} node color choices, and we may bound $\Exp{(X-\mu_X)^{100}}$ assuming full independence (this is a standard argument used in bounded-independent concentration bounds such as that of \cite{BR94}). By Lemma A.1 of \cite{BR94}, therefore,

\begin{align*}
\Exp{(X-\mu_X)^{100}} &\le 2e^\frac{1}{600}\sqrt{100\pi}\left(\frac{100|S|}{e}\right)^{50}
<10^{80} d_{G_0}(v)^{20}\enspace.
\end{align*}

\end{proof}

Our next step is to try to bound the number of monochromatic edges from $s'$ to outside $S$. Denote by $\mathcal E$ the set of edges $\{u,w\}$ with $u\in S$ and $w\in Nom_u \setminus S$. For such an edge $\{u,w\}$, let the variable $Y_{\{u,w\}}$ be defined as follows:
\[Y_{\{u,w\}} = \begin{cases}
1  &\text{if $c(u)=c(w)$ and $w_u(c(u))\le d_{G_0}(v)^{0.1}$ (i.e. if $u\in S'$)}\\
0&\text{otherwise}
\end{cases}\enspace.\]

Let $Y:= \sum_{e \in \mathcal E}Y_e$, let $\mu_e = \Exp{Y_e}$, and let $\mu_Y = \Exp{Y}$. $Y$ is then the number of edges from $S'$ to outside of $S$ over which there are coloring conflicts, and we will need an upper bound for this quantity. First, we give an upper bound for its expectation: 

\begin{claim}\label{claim:muY}
When color choices of \textsc{ColorTrial} are determined by a random hash function from
a $100$-wise independent hash family $[n^{O(1)}] \rightarrow [n^{O(1)}]$, $\mu_Y \le 0.26 d_{G_0}(v)^{0.4}$.
\end{claim}

\begin{proof}
Notice that for any node $u\in S$, we have $u\in P_v$, and so $u$ must have been successful in the nomination step, and therefore has $|Nom_u| \le \frac 14 d_{G_0}(u)+ p_{G_0}(u)^{0.7} \le 0.26p_{G_0}(u)$ (we may assume $p_{G_0}(u)$ is at least a sufficiently large constant $C$ since $p_{G_0}(u)\ge d_{G_0}(u)\ge d_{G_0}(v)$). Then,
	
\begin{align*}
	\mu_Y &= \sum_{u\in S}\sum_{w\in Nom_u \setminus S} \Prob{\text{$c(u)=c(w)$ and $w_u(c(u))\le d_{G_0}(v)^{0.1}$}}\\
	&\le \sum_{u\in S}\sum_{v\in Nom_u } \Prob{\text{$c(u)=c(w)$}}\\
	&\le \sum_{u\in S}\sum_{v\in Nom_u } p_{G_0}(u)^{-1}\\
	&\le \sum_{u\in S}0.26p_{G_0}(u) \cdot p_{G_0}(u)^{-1}\\
	&=0.26d_{G_0}(v)^{0.4}\enspace.
\end{align*}
\end{proof}

Next, we show that $Y$ is concentrated around its expectation:

\begin{claim}\label{claim:Yconc}
When color choices of \textsc{ColorTrial} are determined by a random hash function from
a $100$-wise independent hash family $[n^{O(1)}] \rightarrow [n^{O(1)}]$,	$\Exp{(Y-\mu_Y)^{50}}\le d_{G_0}(v)^{16}$.
\end{claim}

\begin{proof}
By linearity of expectation,

\begin{align*}
\Exp{(Y-\mu_Y)^{50}}& = \Exp{\left(\sum_{e\in \mathcal E} Y_e - \mu_e \right)^{50}}\\
&= \Exp{\sum_{e_1, e_2, \dots, e_{50}\in \mathcal E} (Y_{e_1} - \mu_{e_1})  (Y_{e_2} - \mu_{e_2}) \dots  (Y_{e_{50}} - \mu_{e_{50}})}\\
&= \sum_{e_1, e_2, \dots, e_{50}\in \mathcal E}\Exp{ (Y_{e_1} - \mu_{e_1}) (Y_{e_2} - \mu_{e_2}) \dots (Y_{e_{50}} - \mu_{e_{50}})}\enspace.
\end{align*}

Again, notice that $\Exp{(Y-\mu_Y)^{50}}$ is the sum of expectations of terms involving at most $50$ edges (including duplicates), and so at most $100$ nodes. Each of these sets of $100$ nodes behave independently under a hash function from a $100$-wise independent family. Therefore, the value of $\Exp{(Y-\mu_Y)^{50}}$ is the same as under \emph{fully independent} node color choices, and we may bound $\Exp{(Y-\mu_Y)^{50}}$ assuming full independence of node choices.

We will now fix color choices for all nodes in $S$ (which, since we can now assume independence, will not affect the distributions of other nodes' color choices, for the purposes of bounding $\Exp{(Y-\mu_Y)^{50}}$). 

Since our subsequent analysis will consider `revealing' the random choices in two stages, first for nodes in $S$ and then for nodes outside $S$, we introduce notation to keep track of which random choices are under consideration. Let $\Expu{S}{}$ denote expectation over color choices of nodes in $S$, $\Expu{\bar S}{}$ denote expectation over all other node choices, and $\Exp{}$ denote expectation over all random choices.

Let $\mathcal A$ denote the set of possible color assignments to the nodes of $S$.\hide{\gopi{May be nodes of $S$?} Yes - thanks} For an arbitrary fixed assignment $A\in \mathcal A$, let $\mu^A_Y = \Expu{\bar S}{Y|A}$.

\begin{claim}\label{claim:muAY}
When color choices of \textsc{ColorTrial} are fully independent, $\Expu{S}{(\mu^A_Y-\mu_Y)^{50}}\le 10^{40} d_{G_0}(v)^{15}$. 

\end{claim}

\begin{proof}
By definition, for fixed assignment $A$ and node $u\in S'$, $\Expu{\bar S}{\sum\limits_{w\in Nom_u \setminus S } Y_{\{u,w\}}|A} = w_u(c_A(u))$, where $c_A(u)$ is $u$'s color under $A$. So, \[\Expu{\bar S}{Y|A} =  \sum_{u \in S'} w_u(c_A(u)) = d_{G_0}(v)^{0.1} \sum_{u \in S} X_{u}|A = d_{G_0}(v)^{0.1} X|A\enspace. \]

Furthermore, $\mu_Y = \Expu{S}{\Expu{\bar S}{Y|A}} = \Expu{S}{d_{G_0}(v)^{0.1} X|A} = d_{G_0}(v)^{0.1}\mu_X$. Then, 

\begin{align*}
\Expu{S}{(\mu^A_Y-\mu_Y)^{50}} &= \Expu{ S}{(d_{G_0}(v)^{0.1} X|A-d_{G_0}(v)^{0.1}\mu_X)^{50}}\\
&\le d_{G_0}(v)^{5} \Expu{ S}{( X-\mu_X)^{50}}\\
&\le d_{G_0}(v)^{5} \sqrt{\Expu{S}{( X-\mu_X)^{100}}}\\
&\le 10^{40} d_{G_0}(v)^{15} &\text{by Claim \ref{claim:S'2}.}
\end{align*}
\end{proof}

Now, we bound $\Expu{\bar S}{(Y|A - \mu^A_Y)^{50}}$. Let $\mathcal E'\subseteq \mathcal E$ denote the set $ \{\{u,w\}: c_A(u)\text{ is unique in $S$, i.e. $u\notin Q$}\}$. The important observation is that, once colors $c_A(u)$ for nodes $u\in S$ are fixed, the variables $\{Y_{e}: e\in \mathcal E'\}$ are \emph{negatively associated}:

\begin{definition}[Definition 1 of \cite{DR96}]
Let $X := (X_1,\dots,X_n)$ be a vector of random variables. The random variables X are negatively associated if for every two disjoint index sets, $I,J \subseteq [n]$,

\[\Exp{f(X_i , i \in I)g(X_j , j \in J)} \le \Exp{f(X_i , i \in I)}\Exp{g(X_j , j \in J)}\]

for all functions $f : \mathbb R^{|I|} \rightarrow \mathbb R$ and $g : \mathbb R^{|J|} \rightarrow \mathbb R$ that are both non–decreasing or both non–increasing.
\end{definition}

Negative association of $\{Y_{\{u,w\}}: c_A(u)\text{ is unique in $S$}\}$ can be shown from the definition, but is more easily seen by the combination of the following two lemmas:

\begin{lemma}[Lemma 8 of \cite{DR96}] If $X_1,...,X_n$ are binary variables such that $\sum_{i} X_i \le 1$, then $X_1,...,X_n$ are negatively associated.
\end{lemma} 

(The lemma in fact states $\sum_{i} X_i = 1$, but it can easily be seen from the proof that is still applies when $\sum_{i} X_i \le 1$.)

\begin{lemma}[From Proposition 7 of \cite{DR96}] If collections of variables $X_1,...,X_n$ and $Y_1,...,Y_m$ are individually negatively associated and mutually independent, then variables $X_1,...,X_n, Y_1,...,Y_m$  are negatively associated.
\end{lemma} 

In our case, for any particular $w\notin S$, the set of variables $\{Y_{\{u,w\}}: c_A(u)\text{ is unique in $S$}\}$ sum to at most $1$, since $w$ can take the color of at most one such $u$ (and this is the purpose of requiring that $c_A(u)$ is unique). Then, for every $w' \neq w$, all pairs of variables $Y_{\{u,w\}}$ and $Y_{\{u',w'\}}$ are independent, so we obtain negative association for all of $\{Y_{e}: e\in \mathcal E'\}$.

Let $Y' := \sum_{e\in \mathcal E'}Y_e$. By Proposition 5 of \cite{DR96}, Chernoff/Hoeffding bounds apply to sums of negatively associated variables in in the same way as independent variables, and this includes Lemma A.4 of \cite{BR94} (which can easily be seen by its proof). By this lemma,

\begin{align*}
	\Expu{\bar S}{(Y'-\Expu{\bar S}{Y'})^{50}} &\le 2e^\frac{1}{300}\sqrt{50\pi}\left(\frac{5}{2e}\right)^{25}\cdot (50 (\Expu{\bar S}{Y'} + 50))^{25} \\
	&<10^{65} (\Expu{\bar S}{Y'}+50)^{25}\\
	&\le 10^{65} ( \mu^A_Y+50)^{25}\enspace.
\end{align*}

Now we bound the number of monochromatic edges in $\mathcal E \setminus \mathcal E'$. If $u\notin S'$, then the variables $Y_{\{u,w\}}$ are $0$ by definition. Otherwise, for fixed $u\in Q\cap S'$, the variables $Y_{\{u,w\}}:w\in Nom_u\setminus S $ are independent, and $\Expu{\bar S}{\sum_{w\in Nom_u\setminus S} Y_{\{u,w\}}} = w_u(c(u))$. Denote $Y^u:= \sum_{w\in Nom_u\setminus S} Y_{\{u,w\}}$.

By Lemma A.4 of \cite{BR94}, 

\begin{align*}
	\Expu{\bar S}{(Y^u-\Expu{\bar S}{Y^u})^{50}} &\le 2e^\frac{1}{300}\sqrt{50\pi}\left(\frac{5}{2e}\right)^{25}\cdot (50 (w_u(c(u)) + 50))^{25}
	<10^{65} (w_u(c(u))+50)^{25} \enspace.
\end{align*}

Since $u\in S'$, $w_u(c(u))\le d_{G_0}(v)^{0.1}$, so 

\[\Expu{\bar S}{(Y^u-\Expu{\bar S}{Y^u})^{50}}<10^{65} (d_{G_0}(v)^{0.1}+50)^{25} \le d_{G_0}(v)^{3},\]

since we can again assume $d_{G_0}(v)>C$ for sufficiently large constant $C$.

Then, we can bound the concentration over all such nodes $u\in Q$, using the inequality $(x_1+\dots + x_q)^{50} \le q^{50} \max_{i\le q} x_i^{50}$, which can be seen by expanding the $q^{50}$ terms of $(x_1+\dots + x_q)^{50}$ and noting that each is at most $\max_{i\le q} x_i^{50}$.

\[
\Expu{\bar S}{\left(\sum_{u\in Q}|Y^u-\Expu{\bar S}{Y^u}| \right)^{50}}\le \Expu{\bar S}{|Q|^{50} \cdot \max_{u\in Q}(Y^u-\Expu{\bar S}{Y^u})^{50}}\le  q^{50}d_{G_0}(v)^{3}\enspace.
\]

So, incorporating the concentration of $Y_e$ variables for nodes both in and outside $Q$,

\begin{align*}
\Expu{\bar S}{(Y|A - \mu^A_Y)^{50}}	&\le \Expu{\bar S}{\left(|Y'-\Expu{\bar S}{Y'}|+\sum_{u\in Q}|Y^u-\Expu{\bar S}{Y^u}| \right)^{50}}\\
&\le 2^{50} \cdot\max{\left\{\Expu{\bar S}{(Y'-\Expu{\bar S}{Y'})^{50}},\Expu{\bar S}{\left(\sum_{u\in Q}|Y^u-\Expu{\bar S}{Y^u}| \right)^{50}}\right\}} \\
&< 2^{50} \cdot ( 10^{65} ( \mu^A_Y+50)^{25} + q^{50}d_{G_0}(v)^{3})
\end{align*}

Now we can obtain a concentration bound for $Y$.

\begin{align*}
\Exp{(Y-\mu_Y)^{50}}&= \Expu{S}{\Expu{\bar S}{(Y-\mu_Y)^{50}}}\\
&\le \Expu{S}{\Expu{\bar S}{\left(|\mu^A_Y-\mu_Y| + |Y - \mu^A_Y|\right)^{50} }}\\
&\le 2^{50}\Expu{S}{\Expu{\bar S}{(\mu^A_Y-\mu_Y)^{50} + (Y - \mu^A_Y)^{50} }} &\text{using $(x_1+x_2)^{50}\le 2^{50} (x_1^{50}+x_2^{50})$}\\
&= 2^{50}\Expu{S}{(\mu^A_Y-\mu_Y)^{50} +  \Expu{\bar S}{(Y - \mu^A_Y)^{50} }}\\
&\le 2^{50} 10^{40} d_{G_0}(v)^{15} +\Expu{S}{ 2^{50} \cdot ( 10^{65} ( \mu^A_Y+50)^{25} + q^{50}d_{G_0}(v)^{3})}\\
&\le 2^{50} 10^{40} d_{G_0}(v)^{15} +2^{50}( \Expu{S}{ 10^{65} ( \mu^A_Y+50)^{25} }+d_{G_0}(v)^{3})\enspace.
\end{align*}

Since $\Expu{S}{(\mu^A_Y-\mu_Y)^{50}}\le 10^{40} d_{G_0}(v)^{15}$ by Claim \ref{claim:muAY}, we have 
\begin{align*}
\Expu{S}{( \mu^A_Y+50)^{25}} &\le 2^{25}(\Expu{S}{(\mu^A_Y)^{25}}+50^{25})\\
& \le 2^{25}(\mu_Y^{25}+10^{20} d_{G_0}(v)^{7.5} + 50^{25})\\
&\le 2^{25}((0.26d_{G_0}(v)^{0.4})^{25}+10^{20} d_{G_0}(v)^{7.5} + 50^{25})\\
&\le d_{G_0}(v)^{11}\enspace,
\end{align*}

using that $d_{G_0}(v)$ is at least a sufficiently large constant $C$. This leaves:

\begin{align*}
	\Exp{(Y-\mu_Y)^{50}}&\le 2^{50} 10^{40} d_{G_0}(v)^{15} + 2^{50} \cdot ( 10^{65}  d_{G_0}(v)^{11} + d_{G_0}(v)^{3}) \le d_{G_0}(v)^{16}\enspace.\\
\end{align*}

The proof of \Cref{claim:Yconc} is therefore complete.
\end{proof}

\begin{lemma}\label{lem:S}
When color choices of \textsc{ColorTrial} are determined by a random hash function from
a $100$-wise independent hash family $[n^{O(1)}] \rightarrow [n^{O(1)}]$, with probability at least $1-d_{G_0}(v)^{-3.5}$, at most $0.28d_{G_0}(v)^{0.4}$ nodes in $S$ fail to color themselves.
\end{lemma}

\begin{proof}
We distinguish three ways in which a node $u$ in $S$ can fail to color themselves:

\begin{itemize}
	\item They can choose color $c(u)$ such that $w_{u}(c(u)) > d_{G_0}(v)^{0.1}$, i.e. $u\notin S'$;
	\item They can choose the same color as another node in $S$, i.e. $u\in Q$;
	\item They can be in $S'$ but have a monochromatic edge to a neighbor in $Nom_u\setminus S$.
\end{itemize}

We bound the number of nodes that can fall under each of these cases:

\begin{itemize}
	\item By Claim \ref{claim:S'}, with probability $1-d_{G_0}^{-5}$, $|S\setminus S'|\le 0.27d_{G_0}^{0.35}\enspace.$
	\item By Claim \ref{claim:100}, $\Expu{S}{|Q|^{20}}<1$. By Markov's inequality, 
	\[\Prob{|Q|\ge d_{G_0}(v)^{0.1}}= \Prob{|Q|^{50}\ge d_{G_0}(v)^{5}}\le d_{G_0}(v)^{-5}\enspace.\]
	\item By Claim \ref{claim:Yconc}, $\Exp{(Y-\mu_Y)^{50}}\le d_{G_0}(v)^{16}$. $\mu_Y\le 0.26d_{G_0}(v)^{0.4}$ by Claim \ref{claim:muY}, and so 
	\begin{align*}
		\Prob{Y\ge 0.27d_{G_0}(v)^{0.4}} &\le \Prob{|Y-\mu_Y|\ge 0.01d_{G_0}(v)^{0.4}}\\
		&=\Prob{(Y-\mu_Y)^{50}\ge 0.01^{50}d_{G_0}(v)^{20}}\\
		&< d_{G_0}(v)^{-3.9}\enspace.
	\end{align*}
\end{itemize}
	So, with probability at least $1-d_{G_0}(v)^{-3.9}$, there are at most $0.27d_{G_0}(v)^{0.4}$ monochromatic edges between $S'$ and $V\setminus S$, and so at most $0.27d_{G_0}(v)^{0.4}$ nodes $u$ with such an adjacent monochromatic edge.

Overall, with probability at least $1-(d_{G_0}(v)^{-5}+d_{G_0}(v)^{-5}+d_{G_0}(v)^{-3.9})\ge 1-d_{G_0}(v)^{-3.5}$,\hide{\gopi{May be $d_{G_0}(v)$?} Yes, it should be - thanks} the number of nodes $u$ in $S$ that fail to color themselves is at most $0.27d_{G_0}^{0.35}+d_{G_0}(v)^{0.1}+0.27d_{G_0}(v)^{0.4}\le 0.28d_{G_0}(v)^{0.4}$.
\end{proof}

Now, we can bound the number of nodes in $S$ that choose the same color as a neighbor, and thereby fail to color themselves.

\begin{lemma}\label{lem:failcolor}
When color choices of \textsc{ColorTrial} are determined by a random hash function from
a $100$-wise independent hash family $[n^{O(1)}] \rightarrow [n^{O(1)}]$, with probability at least $1-d_{G_0}(v)^{-3}$, at most $0.29|P_v|$ nodes in $P_v$ fail to color themselves.
\end{lemma}

\begin{proof}
Recall that $S$ was an arbitrary subset of $P_v$ of size $d_{G_0}(v)^{0.4}$. We can divide $P_v$ into $\lceil |P_v|d_{G_0}(v)^{-0.4}\rceil$ such (possibly overlapping) sets $S$. By Lemma \ref{lem:S}, at most $0.28d_{G_0}(v)^{0.4}$ nodes in each such set fail to color themselves. So, by a union bound, with probability at least $1-d_{G_0}(v)^{-3}$, the total number of nodes in $P_v$ that fail to color themselves is at most $\lceil |P_v|d_{G_0}(v)^{-0.4}\rceil \cdot 0.28d_{G_0}(v)^{0.4} \le  0.28(|P_v|+ d_{G_0}(v)^{0.4}) \le 0.29|P_v|$. 
\end{proof}

We can also obtain an upper bound on the number of $v$'s neighbors in $P_v$ that choose a color from $v$'s palette. This comes from the fact that the nodes in $P_v$ have significantly higher degree than $v$, and $v$'s palette must not be much larger than its degree since otherwise it is already successful under \cref{def:fail-col}.

\begin{lemma}\label{lem:samecolor}
When color choices of \textsc{ColorTrial} are determined by a random hash function from
a $100$-wise independent hash family $[n^{O(1)}] \rightarrow [n^{O(1)}]$, with probability at least $1-d_{G_0}(v)^{-4}$, at most $0.41|P_v|$ nodes in $P_v$ choose a color in $v$'s palette.
\end{lemma}

\begin{proof}
 Let $Z_u$ be an indicator variable for the event that $u$ chooses a color in $v$'s palette, and denote $\mu_Z = \Exp{\sum_{u\in P_v}Z_u}$. By the definition of  $\mathrm{Nom}^*_v$, $u$  has palette size greater than $3d_{G_0}(v)$, and therefore has at least $p_{G_0}(u) - p_{G_0}(v) \geq 0.6p_{G_0}(u)$ colors not in $v$'s palette. So, $\Prob{Z_u = 1}\le 0.4$, and $\mu_Z\le 0.4|P_v|$. Furthermore, the variables $Z_u$ are $100$-wise independent when color choices are chosen by a random hash function from a $100$-wise independent family. By Lemma 2.2 of \cite{BR94}, 
 
 \begin{align*}
 	\Prob{|Z-\mu_Z|\ge |P_v|^{0.6}} &\le 8\left(\frac{50\mu_Z+10000}{|P_v|^{1.2}}\right)^{50}\\
 	&\le 8\left(\frac{20|P_v|+10000}{|P_v|^{1.2}}\right)^{50}\\
 	&< 8\left(\frac{0.8 d_{G_0}(v)+10000}{(0.04 d_{G_0}(v))^{1.2}}\right)^{50} \\
 	&\le 8\left(d_{G_0}(v)^{-0.1}\right)^{50} \\
 	&\le d_{G_0}(v)^{-4}\enspace.
 \end{align*}
 
 So, with probability at least $1-d_{G_0}(v)^{-4}$, the number of nodes in $P_v$ that choose a color from $v$'s palette is at most $0.4|P_v|+|P_v|^{0.6} \le 0.41|P_v|$.
\end{proof}

Lemmas \ref{lem:failcolor} and \ref{lem:samecolor} together bound the number of nodes in $P_v$ that do not successfully color themselves using a color not in $V$'s palette:
	
\begin{lemma}
	\label{lem:colorstep}
	When color choices in the coloring step of \textsc{ColorTrial} are determined by a random hash function from a $100$-wise independent hash family $[n^{O(1)}]\rightarrow [n^{O(1)}]$, any node $v$ that did not fail in the nomination step fails in the coloring step of \textsc{ColorTrial} with probability at most $1/p_{G_0}(v)$.
\end{lemma}

\begin{proof}
Of nodes in $P_v$, at most $0.41|P_v|$ choose a color in $v$'s palette with probability at least $1-d_{G_0}(v)^{-4}$, by Lemma \ref{lem:samecolor}. At most a further $0.29|P_v|$ fail to color themselves due to monochromatic edges, with probability at least $1-d_{G_0}(v)^{-3}$, by Lemma \ref{lem:failcolor}. This leaves at least $0.3|P_v|$ nodes that successfully color themselves a color not in $v$'s palette. $0.3|P_v|> 0.012 d_{G_0}(v) \le  0.01 p_{G_0}(v)$, since (assuming $v$ is not already successful by the first point of Definition \ref{def:fail-col}) $p_{G_0}(v) <  1.1 d_{G_0}(v)$, and therefore $v$ is successful in this case. The probability that $v$ is \emph{not} successful is therefore at most $d_{G_0}(v)^{-4}+d_{G_0}(v)^{-3} \le d_{G_0}(v)^{-2} \le p_{G_0}(v)^{-1}$.
		
\end{proof}
		
The induced graph of failed nodes is therefore small when randomness is chosen from a family of $100$-wise independent hash functions, and by the method of conditional expectations we can obtain the same deterministically.
\begin{lemma}
\label{lem:MCE2}
We can deterministically choose a hash function in $O(1)$ rounds, from a $100$-wise independent family $[n^{O(1)}]\rightarrow [n^{O(1)}]$, to run the coloring step of \textsc{ColorTrial} such that the size of the subgraph induced by the failed nodes (in the coloring step) is at most $n$.
\end{lemma}
\begin{proof}
Let $P \subseteq V(G_0)$ be the set of nodes that are successful during the nomination step. For each $v$ in $P$, let us define a indicator random variable $\mathbf{1}_v$ such that it takes value $1$ if $v$ is a failed node in the coloring step, and 0 otherwise. Let us define another random variable $I_v=\mathbf{1}_v \cdot d_{G_0}(v)$ for each node  $v$ in $G_0$. Our local cost function $f(h,v)$ for the method of conditional expectations takes the value of $I_v$ when coloring choices are specified by $h$, and the global cost function $F(h) = \sum_{v \in P}f(h,v)$.  Observe that the size of the required subgraph (induced by the failed nodes in the coloring step) is bounded by $F(h) = \sum_{v \in P}I_v$. By \Cref{lem:nominations} and by the linearity of expectation, 
the expected size of the required subgraph is at most $\sum_{\text{$v\in P$}}d_{G_0}(v) \cdot 1/p_{G_0}(v) \le n$. Note that $\sum_{v \in P}I_v$ is  the aggregate of $I_v$ variables. Each  node $v$ can compute $I_v$ by checking the conditions mentioned in \Cref{def:fail-col} if the $2$-hop neighborhood of a node $v$ is known. As we assume $\Delta_G \leq O(\sqrt{n})$ in \textsc{Color} (\Cref{alg:mainrand}), $2$-hop neighborhoods of nodes can be collected to their network nodes in $O(1)$ rounds of \CONGESTEDC. The method of conditional expectations using $F(h)$ (outlined in \Cref{subsec:method_of_conditional_expectations}) implies that we can find a hash function deterministically so that the size of the subgraph induced by the failed nodes (in the coloring step) is at most $n$.
%
\end{proof}

Note that \Cref{lem:MCE2} is the only lemma whose proof requires the assumption $\Delta_G=O(\sqrt{n})$.

Recall that $F$ denotes the subgraph of $G$ induced by the failed nodes  either in the nomination step or in the coloring step of \textsc{ColorTrial}. The following lemma bounds $|F|$, and follows immediately from \Cref{lem:MCE} and \Cref{lem:MCE2}.

\begin{lemma}
\label{lem:sizeF}
We can deterministically choose hash functions in $O(1)$ rounds, from a $100$-wise independent hash family $[n^{O(1)}]\rightarrow [n^{O(1)}]$, to run each step of \textsc{ColorTrial} such that the size of the subgraph induced by the failed nodes (either in the nomination step or in the coloring step) is at most $O(n)$, i.e., $\size{F}=O(n)$.
\end{lemma}

\subsection{\textsc{SubSample}}\label{sec:samp}

After executing \textsc{ColorTrial}$(G_0)$, \textsc{Color}$(G,x)$ executes  procedure \textsc{Subsample$(G_1,x)$}. The randomized procedure on which \textsc{SubSample} is based is Algorithm~\ref{alg:subsamplerand}. To derandomize \textsc{SubSample}, we replace the random choice of line 1 (to generate a set $S$ of vertices) with a choice determined by a hash function from a $100$-wise independent family $[n^{O(1)}]\rightarrow [n^{O(1)}]$.

\begin{algorithm}[H]
	\caption{\textsc{Subsample$(G_1,x)$} - Randomized Basis}\label{alg:subsamplerand}

	Each node $v$ in $G_1 $ independently joins $S$ with probability $d_{G_1}(v)^{-0.1}$

        Each node $v$ in $G_1$ decides whether it succeeds or fails. Let $F_1$ be the set of failed nodes.

        Let $L$ denote the nodes with $p_{G_1}(v)<C$.

	Return:
	\begin{itemize}
        \item $G' := (S\cup F_1 \cup L)$
		\item $G_2 := G_1\setminus G'$
		
	\end{itemize}
\end{algorithm}

Note that while $x$ is not used explicitly in \Cref{alg:subsamplerand}, it increases by $0.1$ in each recursive call to \textsc{Color}, and is part of the objective function used for the derandomization of \Cref{alg:subsamplerand}. The purpose of $x$ is in the analysis to show some measure of the `size' of $G'$ must decrease by $0.1$ in the exponent in each recursive call, and so after $10$ levels of recursion of \textsc{Color}$(G, 0)$, the remaining graph is of size $O(n)$ (see \Cref{lem:G'size}). This value is $0.1$ must necessarily be a small constant for the subsampling to provide sufficient `slack' to be used in \textsc{BucketColor}.


We start by defining the notion of failed nodes in \textsc{SubSample}.

\begin{definition}\label{defi:fail-samp}
Let us define $N^\approx(v)\subseteq N_{G_1}(v)$ to be the subset of $v$'s neighbors $u$ with $\frac12 d_{G_1}(v)\le d_{G_1}(u)\le 6 d_{G_1}(v)$.
A node $v$ is classed as \textbf{successful} during \textsc{SubSample} if either
\begin{enumerate}
\item Large palette size: $p_{G_1}(v) \ge 1.1 d_{G_1}(v)$; or
\item Few neighbors with roughly equal degree: $|N^\approx(v)| \le \frac 13 d_{G_1}(v)$; or
\item Large number of neighbors joining $S$: at least $\frac14 p_{G_1}(v)^{0.9}$ of $v$'s neighbors join $S$.
\end{enumerate}
Node $v$ is classed as \textbf{failed} if none of the above three conditions hold.
\end{definition}
Note that the set $S$ is generated within \textsc{SubSample}. For a node $v$, the first two conditions in the definition above can be verified using the graph $G_1$, which is provided as input to \textsc{SubSample}. However, the third condition depends on both the set $S$ and the graph $G_1$.
.

In a similar way to the analysis in \Cref{sec:trial}, we can analyze \textsc{SubSample} under bounded-independence hash functions and derandomize it, obtaining the following:

\begin{lemma}\label{lem:G'size}
For all $x\in [0,0.9] $, we can deterministically choose a hash function in $O(1)$ rounds, from a $100$-wise independent hash family,  to execute line 1 of \textsc{SubSample} to generate set $S$ such that the following holds:  $\sum_{\text{$v\in G'$}}d_{G_1}(v)^{x+0.1}\le Cn+ 2\sum_{\text{$v\in G_1$}}d_{G_1}(v)^{x} .$
\end{lemma}

To prove the above lemma, we first show in \Cref{lem:subsample} that a node $v$ in $G_1 \setminus L$ fails or joins $S$ during the execution of \textsc{SubSample} with probability at most $2d_{G_1}(v)^{-0.1}$. While proving \Cref{lem:G'size}, we use \Cref{lem:subsample} to argue that the expected size of $\sum_{\text{$v\in G'$}}d_{G_1}(v)^{x+0.1}$ is at most $Cn+ 2\sum_{\text{$v\in G_1$}}d_{G_1}(v)^{x}$.

\begin{lemma}\label{lem:subsample}
Consider line 1 of \textsc{SubSample} where we generate set $S$. When random choices to generate $S$ are determined by a random hash function from a $100$-wise independent family $[n^{O(1)}]\rightarrow [n^{O(1)}]$, any fixed node $v$ in $G_1 \setminus L$ fails or joins $S$ during \textsc{SubSample} with probability at most $2d_{G_1}(v)^{-0.1}$.
\end{lemma}
\begin{proof}
Since the probability of node $v$ joining $S$ is $d_{G_1}(v)^{-0.1}$, we will be done by showing that $v$ fails with probability at most $p_{G_1}(v)^{-1}\le d_{G_1}(v)^{-0.1} $. 

 Recall \Cref{defi:fail-samp} that contains three properties and note that a node $v$ fails in \textsc{SubSample} if none of those three properties hold. We first note that the first two properties of \Cref{defi:fail-samp} are already determined (from \textsc{ColorTrial}) and are not dependent on the behavior of \textsc{SubSample}. So, to compute the probability of the event that $v$ fails, we  analyze the probability of satisfying the third property assuming that the first two do not hold, i.e., $p_{G_1}(v)<1.1d_{G_1}(v)$ and $\size{N^{\approx}(v)}> \frac{1}{3}d_{G_1}(v)$. In particular, we need to show that $\Prob{|{N_{G_1}(v) \cap S}| < \tfrac 14 p_{G_1}(v)^{0.9}} \leq p_{G_1}(v)^{-1}.$ Since $N^{\approx}(v) \subseteq N_{G_1}(v)$,  it is sufficient to argue that $\Prob{|{N^{\approx}(v) \cap S}| < \tfrac 14 p_{G_1}(v)^{0.9}} \leq p_{G_1}(v)^{-1}.$

For $u\in N^\approx(v)$, let $X_u$ be the indicator random variable that takes value $1$ or $0$ depending on whether $u$ joins $S$ or not. So, the number of vertices of $u\in N^\approx(v)$ that joins $S$ is $\size{N^\approx(v) \cap S}=\sum_{u\in N^\approx(v)}X_u.$ Note that $\Prob{X_u=1} = d_{G_1}(u)^{-0.1}$ and the expected value of $\size{N^\approx(v) \cap S}$ is $\mu =\sum_{u \in N^{\approx}(v)} d_{G_1}(u)^{-0.1}$. Since $p_{G_1}(v)<1.1d_{G_1}(v),\size{N^{\approx}(v)}> \frac{1}{3}d_{G_1}(v)$, and  $d_{G_1}(u) \leq 6d_{G_1}(v)$ for each $u \in N^{\approx}(v)$, we can deduce the following:

$$\mu \geq \sum_{u \in N^{\approx}(v)}(6d_{G_1}(v))^{-0.1} > \frac{1}{3}d_{G_1}(v) \cdot (6d_{G_1}(v))^{-0.1}> \frac{1}{3 \cdot 6^{0.1}}(p_{G_1}(v)/1.1)^{0.9} > 0.2557 p_{G_1}(v)^{0.9}.$$ We may assume  that $\mu \geq 1000$ as $p_{G_1}(v) \geq C$ for sufficiently large constant $C$. Moreover, the set of random variables $\{X_u:u\in N^{\approx}(v)\}$ are $100$-wise independent  when the randomness for the procedure is provided by a random hash function from a $100$-wise independent family.


So, applying \Cref{cor:conc} (with $X=|N^{\approx}(v) \cap S|$ and $A= 0.0057 p_{G_1}(v)^{0.9}$), we have the following:

$$ \Prob{\size{N^{\approx}(v) \cap S}\leq \mu-0.0057 p_{G_1}(v)^{0.9}}\leq \left(\frac{111\mu}{(0.0057 p_{G_1}(v)^{0.9})^2}\right)^{50}.$$

We have argued that $\mu > 0.2575 p_{G_1}(v)^{0.9}$. Also, note that $\mu \leq \size{N^{\approx}(v)}\leq N_{G_1}(v)=  d_{G_{1}}(v)< p_{G_1}(v).$

%
\begin{align*}
    \Prob{|{N^{\approx}(v) \cap S}| < \tfrac 14 p_{G_1}(v)}
        &\le
    \left(\frac{111\cdot p_{G_1}(v)}{(0.0057 p_{G_1}(v)^{0.9})^2}\right)^{50}
	\leq 
    \left(3416436 \cdot p_{G_1}(v)^{-0.8}\right)^{50}
	\leq
    p_{G_1}(v)^{-1}
    \enspace.
    \qedhere
\end{align*}
%
\end{proof}

Now we show that the selection of hash function in \Cref{lem:subsample} can be derandomized, hence proving \Cref{lem:G'size}.

\begin{proof}[Proof of \Cref{lem:G'size}]
Recall that $G'=L \cup S \cup F_1$, i.e., a node $v$ in $G_1$ is placed in $G'$ if it has $p_{G_1}(v) < C$, is sampled into $S$, or fails. For each $v$ in $ G_1\setminus  L$, let us define a indicator random variable $\mathbf{1}_v$ if  $v$ fails or joins $S$, and $0$ otherwise. Let us define another random variable $I_v=\mathbf{1}_v \cdot d_{G_1}(v)^{x+0.1}$ for each $v$ in $G_1 \setminus L$. Our local cost function $f(h,v)$ for the method of conditional expectations takes the value of $I_v$ when coloring choices are specified by $h$, and the global cost function $F(h) = \sum_{v \in G_1\setminus  L}f(h,v)$. Observe that the required sum $\sum_{v\in G'} d_{G_1}(v)^{x+0.1}$ can be expressed as $\sum_{v\in L}d_{G_1}(v)^{x+0.1}+ \sum_{{v\in G_1 \setminus L}}I_v.$ By the definition of $L$, $\sum_{v\in L}d_{G_1}(v)^{x+0.1}\leq Cn$. For each node in $G_1 \setminus L$, from \Cref{lem:subsample}, the probability that $v$ fails or joins $S$ is at most $2d_{G_1}(v)^{-0.1}$, i.e., $\Exp{I_v}\leq d_{G_1}(v)^{x+0.1} \cdot 2d_{G_1}(v)^{-0.1}\leq 2d_{G_1}(v)^{x}$. So, by the linearity of expectation, the expected value of $\sum_{\text{$v\in G'$}}d_{G_1}(v)^{x+0.1} $ is at most  $Cn+ 2\sum_{\text{$v\in G_1 \setminus L$}} d_{G_1}(v)^{x} \leq Cn+2\sum_{\text{$v\in G_1 $}} d_{G_1}(v)^{x} $. Note that $\sum_{v \in G_1 \setminus L} I_v$ is an aggregate function of $I_v$'s. Each node $v$ can compute $I_v$ by checking conditions mentioned in \Cref{defi:fail-samp} when one-hop neighborhoods are known. The method of conditional expectations applied to $F(h)$ implies that we can find a hash function deterministically to execute line 1 of \textsc{SubSample} to generate set $S$ such that $\sum_{\text{$v\in G'$}}d_{G_1}(v)^{x+0.1}\le Cn+ 2\sum_{\text{$v\in G_1$}}d_{G_1}(v)^{x} $.
%
\end{proof}

\subsection{Properties of $G_2$}\label{ssec:G-2}

Here we give a lemma (\Cref{lem:properties}) specifying properties of the graph $G_2$.
Recall that $G_2$ is the graph of successful nodes that results from running \textsc{ColorTrial} and \textsc{SubSample} on our input graph, and it is the input graph to our main coloring procedure \textsc{BucketColor} in \Cref{sec:buck}. Here, we show
that each node in $G_2$ has sufficient slack to be colored in $O(1)$ rounds by \textsc{BucketColor}.

\begin{lemma}
\label{lem:properties}
For any $v\in G_2$, either $p_{G_2}(v) \ge  d_{G_2}(v)+\frac 14 d_{G_2}(v)^{0.9} $, or $|N_{G_2}^-(v)|\ge \frac 13d_{G_2}(v)$.
\end{lemma}

\begin{proof}
Any node $v$ in $G_2$ must have been successful during both steps of \textsc{ColorTrial} and during \textsc{SubSample}, and also must not have been placed in $S$ or $L$ during \textsc{SubSample}. Since $v$ was not placed in $L$ during \textsc{SubSample}, we have $p_{G_1}(v)\ge C$. Note that palettes are not affected during \textsc{SubSample}, so $p_{G_2}(v)=p_{G_1}(v)\ge C$. Also, we may assume that $d_{G_2}(v)$ is at least  a sufficiently large constant; otherwise  $p_{G_2}(v)  \ge d_{G_2}(v)+\frac 14 d_{G_2}(v)^{0.9}$ holds (and we are done) since $p_{G_2}(v) \geq C$ for sufficiently large constant $C$.

%

We first show that, after \textsc{ColorTrial}, either $p_{G_1}(v)\ge 1.01 d_{G_1}(v)$, or $N^{\gtrsim}(v)<\frac 13 d_{G_0}(v)$. This follows from the definition of success (\Cref{def:fail-col}) during the coloring step of \textsc{ColorTrial}: either $N^{\gtrsim}(v)<\frac 13 d_{G_0}(v)$, or one of the three other conditions holds, all of which imply $p_{G_1}(v)\ge 1.01 d_{G_1}(v)$. Note that palettes are not affected during \textsc{SubSample}. If $p_{G_1}(v)\ge 1.01 d_{G_1}(v)$, then $p_{G_2}(v)=p_{G_1}(v) \ge  d_{G_2}(v) +0.01 d_{G_2}(v) \ge d_{G_2}(v)+\frac 14 d_{G_2}(v)^{0.9}.$    
So, we now consider nodes $v$ for which $N^{\gtrsim}(v)<\frac 13 d_{G_0}(v)$.

%
By the definition of success (\Cref{defi:fail-samp}) during \textsc{SubSample}, $v$ has  either has $p_{G_1}(v)\ge 1.1d_{G_1}(v)$ (in which case $p_{G_2}(v) \ge d_{G_2}(v) +\frac 14 d_{G_2}(v)^{0.9} $ by the same argument as above, so we already have sufficient slack and are done), $\frac14 d_{G_1}(v)^{0.9}$ neighbors joining $S$ (each of which generates $1$ slack in $G_2$, so we are again done), or $|N^\approx(v)| \le \frac 13 d_{G_1}(v)$. So, we now need only consider nodes for which $|N^\approx(v)| \le \frac 13 d_{G_1}(v)$ and $N^{\gtrsim}(v)<\frac 13 d_{G_0}(v)$.

\vspace{5pt}
Now we divide the analysis into two cases based on whether $d_{G_2}(v)> \frac12 d_{G_0}(v)$ or not. We will be done by showing that $|N_{G_2}^-(v)|\geq \frac 13 d_{G_2}(v)$ in the former case and   $p_{G_2}(v)\ge  d_{G_2}(v)+\frac 14 d_{G_2}(v)^{0.9}$ in the latter case. The intuition is that, if $v$'s degree has not dropped significantly after \textsc{ColorTrial} (which is effective for higher-degree neighbors) and \textsc{SubSample} (which is effective for similar-degree neighbors), then it must have had many lower-degree neighbors.

\paragraph{Case 1: $d_{G_2}(v) > \frac12 d_{G_0}(v)$.} 
We show that in this case, every neighbor $u\in G_2$ not in $N^{\gtrsim}(v)$ or $N^\approx(v)$ must be in $N_{G_2}^-(v)$, i.e., $d_{G_2}(u)< d_{G_2}(v)$. If $u\notin N^{\gtrsim}(v)$ (i.e., $d_{G_0}(u)< 3 d_{G_0}(v)$), then 
$d_{G_1}(u)\le  d_{G_0}(u)< 3d_{G_0}(v)< 6d_{G_2}(v) \leq 6 d_{G_1}(v).$ If $u \notin N^{\approx}(v)$, either $d_{G_1}(u)<\frac{1}{2}d_{G_1}(v)$ or $d_{G_1}(u)> 6 d_{G_1}(v)$. So, when $u\notin N^\approx(v) \cup N^{\gtrsim}(v)$, we have $d_{G_1}(u)<\frac{1}{2}d_{G_1}(v)$. This implies 
 $d_{G_2}(u) \le d_{G_1}(u) < \frac 12 d_{G_1}(v)\le \frac 12 d_{G_0}(v)< d_{G_2}(v)$, i.e., $u \in N_{G_2}^-(v)$. Hence,
\begin{align*}
    |N_{G_2}^-(v)| 
    &\ge
    |N_{G_2}(v)|-|N^{\approx}(v)|-|N^{\gtrsim}(v)| 
    \ge
    d_{G_2}(v) - \frac 13 d_{G_1}(v) - \frac 13 d_{G_0}(v) \ge \frac 13 d_{G_2}(v)
    \enspace.
\end{align*}

\paragraph{Case 2: $d_{G_2}(v)\le \frac12 d_{G_0}(v)$.} 
First, we give the following claim that bounds the palette size of nodes in $G_1$ in terms of palette size of nodes in $G_0$. 


\begin{claim}\label{lem:palettebound}
For all nodes $v\in G_1$, $p_{G_1}(v) \ge \frac 23 p_{G_0}(v)$.
\end{claim}

\begin{proof}
For a node $v$ to be in $G_1$, it must have been successful during \textsc{ColorTrial}. In particular, it must have had $|\mathrm{Nom}_v|\le \frac 14 d_{G_0}(v)+ p_{G_0}(v)^{0.7}$, since that was one of the success properties during the nomination step. The only way $v$ can lose some color $c$ from its palette between $G_0$ and $G_1$ is by having a nominated neighbor successfully coloring itself $c$. So, $p_{G_1}(v) \ge p_{G_0}(v)-  \frac 14 d_{G_0}(v)- p_{G_0}(v)^{0.7}$. Since we had $p_{G_0}(v)\ge C$ for sufficiently large constant $C$, we therefore get $p_{G_1}(v) \ge \frac 23 p_{G_0}(v)$.
\end{proof}

By \Cref{cl:cru} and $d_{G_2}(v)\le \frac12 d_{G_0}(v)$, we have 
$$
p_{G_1}(v) \ge \frac 23 p_{G_0}(v) \ge \frac 23 d_{G_0}(v) \ge \frac 43 d_{G_2}(v) \ge  d_{G_2}(v)+\frac 14 d_{G_2}(v)^{0.9}.
$$
Since palettes are not affected during \textsc{SubSample},  $p_{G_2}(v)= p_{G_1}(v) \geq  d_{G_2}(v)+\frac 14 d_{G_2}(v)^{0.9}.$
%
%
\end{proof}

\input{parts/bucket.tex} 

%% file: parts/bucket.tex
\section{\textsc{BucketColor}}
\label{sec:buck}

In this section, we describe our core coloring procedure \textsc{BucketColor}$(G_2)$. Note that each node in the input graph $G_2$ to \textsc{BucketColor} has sufficient slack by \Cref{lem:properties}. Throughout this section, the graph under consideration is always $G_2$, so we omit the subscript $G_2$ from $N_{G_2}(v), d_{G_2}(v), N^+_{G_2}(v), d^+_{G_2}(v),N^-_{G_2}(v),$ $d^-_{G_2}(v), \Psi_{G_2}(v)$, $p_{G_2}(v)$ and $\Delta_{G_2}$.

In \Cref{subsec:bucket-structure}, we first formalize the bucket structure of nodes (as discussed in \Cref{subsec:techniques}), and then introduce some useful definitions. Then we describe algorithm \textsc{BucketColor} in \Cref{subsec:bucket-structure}. In Sections \ref{sec:color-bad}--\ref{sec:color-good}, we analyze the correctness of \textsc{BucketColor}.


\subsection{Assigning Nodes to Buckets}
\label{subsec:bucket-structure}

We use two special functions in the description of our algorithm in this section: $l: V(G_2) \rightarrow \mathbb{N}_{\geq 0}$ and $b:\mathbb{N} _{\geq 0}\rightarrow \mathbb{N}_{\geq 0}$. Here $l$ is defined as $l(v):=\max\{\lfloor \log _{1.1} \log _2 d(v)\rfloor,0\}$ for node $v$, and $b$ is  defined as  $b(i) :=  \lfloor 0.7\cdot 1.1^i \rfloor$ for $i \in \mathbb{N}_{\geq 0}$. If $d(v)$ is at least a suitable constant, then $b(l(v)) = \Theta( \log  d(v))$ and $b(l(v)) \leq 0.7 \log _2 d(v)$.

We consider a partition of the nodes of $G_2$ into $O(\log\log\Delta)$ levels, with the \emph{level} of a node $v$ equal to $l(v) = \max\{\lfloor \log _{1.1} \log _2 d(v)\rfloor,0\}$. The nodes of a particular level are further partitioned into buckets. The \emph{level of a bucket $x$} is the level of a candidate node that can be put into this bucket, and is denoted by $level(x)$. The buckets of $level$ $i$ (or $level$-$i$ buckets)  are identified by binary strings of length $b(i)$, where $i \in \mathbb{N}_{\geq 0}$, as well as their level. \footnote{Note that, at low levels, buckets in different levels can be identified by the same string, because the function $b(i)= \lfloor 0.7\cdot 1.1^i \rfloor$ is not injective for $i\le 24$. Therefore, for example, $b(0) = b(1) = 0$, and so levels $0$ and $1$ both in fact contain a single bucket specified by the empty string. We treat these as different buckets in order to conform to a standard rooted tree structure, and therefore must identify buckets by their level as well as their specifying string.} So, there are $2^{b(i)}$ $level$-$i$ buckets. To put a node $v$ into a bucket, (in our algorithm) we
generate a random binary string of length $b(l(v))$.

The set of buckets forms a hierarchical tree structure as described below. We say that a bucket $a'$ is a \emph{child} of $a$ (and $a$ is the \emph{parent} of $a'$) if $level(a') = level(a)+1$ and $a \sqsubseteq a'$. We say that $a'$ is a \emph{descendant} of $a$ (and $a$ is an \emph{ancestor} of $a'$) if $level(a') \ge level(a)$ and $a \sqsubseteq a'$ (note that by definition $a$ is a descendant and ancestor of itself). The buckets form a rooted tree structure: the root is the single level $0$ bucket, specified by the empty string; each bucket in level $i>0$ has one parent in level $i-1$ and multiple children in level $i+1$.

We also put colors into the buckets. For every color $c$, we put $c$ into a bucket $a$ of level $\lfloor \log _{1.1} \log _2 \Delta \rfloor$ by assigning it a binary string of length $b(\lfloor \log _{1.1} \log _2 \Delta \rfloor + 20)$. This will be the maximum level we will use in our algorithms and analysis, so these buckets can be considered leaf buckets. There are, therefore, 
\[O(2^{b(\lfloor \log _{1.1} \log _2 \Delta \rfloor + 20))} = O(2^{0.7\cdot1.1^{(\lfloor \log _{1.1} \log _2 \Delta \rfloor + 20)}})= 
O(2^{0.7\cdot \log_2\Delta \cdot 1.1^{20)}}) = O(\Delta^5)\] buckets in total. We also assign $c$ to all buckets on the path from $a$ to the root bucket (i.e., all ancestors of $a$). That is, $c$ is assigned to bucket $a'$ (and that bucket $a'$ contains $c$) iff $a' \sqsubseteq a$. 

Note that our algorithm uses a hash function to generate the binary strings (and hence the buckets) for the colors and nodes. In particular, our algorithm colors each node $u$ in bucket $x$ with a color present in some descendant bucket of $x$. Thus, an edge $\{u, w\}$, where neither $u$ is a descendant of $w$ nor $w$ is a descendant of $u$, can safely be removed from the graph. 

We argue that it suffices to color a set of reduced coloring instances, one per bucket of the original \DILC instance. For a bucket $x$, the corresponding reduced instance graph contains all the edges where one endpoint lies in bucket $x$ and the other in a descendant bucket of $x$. Additionally, the palette for a node in bucket $x$ is restricted to the colors present in the descendant buckets of $x$.

\begin{definition}
\label{def:hash-functions-notation}
Let $h:(\mathcal{C} \cup V(G_2))\rightarrow \{0,1\}^*$ be a function mapping colors and nodes to binary strings. Let $G^+_2$ be a directed graph with the same set of nodes as $G_2$. For each edge $\{u,w\}$ in $G_2$ such that $w$ is in some descendant bucket of $u$, i.e., $h(w) \sqsupseteq h(u)$ and $d(w)\geq d(u)$, do the following:   If $h(w) \sqsupset h(u)$, include $(u,w)$ as a directed edge in $G_2^+$. If $h(w) = h(u)$ (i.e., $u$ and $w$ are in the same bucket) and $d(w) > d(u)$, include $(u,w)$ as a directed edge in $G_2^+$.  If $h(w) = h(u)$ and $d(w) = d(u)$, include in $G_2^+$ whichever of $(u,w)$ or $(w,u)$ is directed toward the higher-ID endpoint.  For $x \in \{0,1\}^*$:
\begin{itemize}
    \item $G_{x}^+$ is the directed graph containing edges $(u,w)$ such that $(u,w)$ is in $G_2^+$ and $h(u) = x$, along with all nodes that are endpoints of such edges.  
    \item For each $v$ such that $h(v) = x$, define $\Psi_{x}(v) = \{c \in \Psi(v) : h(c) \sqsupseteq x \}$, \hfill {\small ($\Psi_{x}(v)$ is the set of palette colors $v$ has in descendant buckets of $x$)}.
\end{itemize}

\end{definition}



For each $v$ in bucket $x$, let $N^+_x(v)$ denote the set of neighbors of $v$ in $G^+_x$, and let $d_x^+(v) = \size{N^+_x(v)}$ denote the degree of $v$ in $G^+_x$. Additionally, as defined, $\Psi_{x}(v)$ represent the color palette of $v$ in the reduced instance, and let $p_{x}(v) = \size{\Psi_{x}(v)}$ denote the size of the color palette of $v$.

 
Notice that each edge $\{u, w\}$ (or $(u,w)$) is present in at most one reduced instance corresponding to a particular bucket. However, a node may appear in the reduced instances of multiple buckets. In such cases, each pair of these buckets satisfies an ancestor-descendant relationship. Thus, the reduced instances $G_{x}^+$ are not necessarily independent. Moreover, they can be of size $\omega(n)$. Additionally, it is possible that $G^+_{x}$ is not a valid \DILC instance. To handle the issue, we define the notion of \emph{bad nodes} in \Cref{defi:bad}. Intuitively, bad nodes are those who do not behave as expected when mapped to their bucket (e.g. have too many neighbors or too few colors therein), and we will show that the subgraphs of buckets restricted to good nodes are of size $O(n)$ and can be colored in $O(1)$ rounds. If we choose our hash function uniformly at random from a $100$-wise independent family of hash functions, the subgraph $G_{bad}$ (induced by the bad nodes) has size $O(n)$ in expectation. We also show that it is possible to choose a hash function deterministically in $O(1)$ rounds such that the size of $G_{bad}$ is indeed $O(n)$.

\begin{definition}
\label{defi:bad}
Given a hash function $h:(\mathcal{C} \cup V(G_2))\rightarrow \{0,1\}^*$ mapping colors and nodes to binary strings, define a node $v$ in $G_2$ to be \textbf{bad} if at least one of the following occur:
\begin{enumerate}
\item[(1)] More higher degree neighbors in the descendant buckets: $d^+_{h(v)}(v) \ge d^+(v) 2^{-b(l(v))} + \frac 18 d(v)^{0.9}2^{-b(l(v))}$;
\item[(2)] Few colors in the descendant buckets: $p_{h(v)}(v)\le p(v) 2^{-b(l(v))}  -\frac 18 d(v)^{0.9}2^{-b(l(v))}$;
\item[(3)] $v$ has at least one level $l(v)+20$ descendant bucket containing more than one of $v$'s palette colors;
\item[(4)] more than $2n2^{-b(l(v))}$ nodes $v'$ have $h(v)=h(v')$, i.e., in the same bucket as of $v$ 
\end{enumerate}
$v$ is said to be a \textbf{good} if it satisfies none of the above four conditions.
\end{definition}

(1) and (2) 
ensure that each reduced instance (after removing the bad nodes) is a valid \DILC instance; (3) ensure that the dependencies among the reduced instances are limited; and (4) when combined with (1) ensures that each $G_x^+$ is of size $O(n)$.

Now we are ready to discuss our algorithm \textsc{BucketColor}. The randomized procedure on which \textsc{BucketColor} is based is Algorithm~\ref{alg:bucketcolor}. Note that only line 1 of Algorithm~\ref{alg:bucketcolor} is a randomized step, and it can be derandomized by replacing its random choices with choices determined by a hash function from a $101$-wise independent hash family $[n^{O(1)}]\rightarrow [n^{O(1)}]$. The subgraph induced by the bad nodes, $G_{bad}$, is deferred to be colored later. Then in Lines 3 to 11, \textsc{BucketColor} colors the (good) nodes in $G_2 \setminus G_{bad}$ in $O(1)$ rounds deterministically.

\begin{algorithm}[h]
\caption{\textsc{BucketColor$(G_2)$} --- Randomized Basis}
\label{alg:bucketcolor}
Each node $v$ uniformly randomly chooses a $b(l(v))$-bit binary string $h(v)$, and each color is uniformly randomly assigned a $b(\lfloor \log _{1.1}\log _2\Delta\rfloor+ 20 )$-bit binary string $h(c)$.
	
Each node $v$ decides whether it is bad or good according to \Cref{defi:bad}. Let $G_{bad}$ be the subgraph induced by the bad nodes.
 
\For{$O(1)$ iterations}
{
		Each node $v\in G_2 \setminus G_{bad}$ restricts its palettes to colors $c$ with $h(v)\sqsubseteq h(c)$, i.e., $\Psi_{h(v)}(v)\gets\{c \in \Psi(v):h(v)\sqsubseteq h(c)\}$.
		
		For each $i \in [\lfloor \log_{1.1} \log _2 \Delta\rfloor + 20]$ and each string $x \in \{0,1\}^{b(i)}$, collect the graph $G_x^+$ to a dedicated network node $node_x$.
		
		For each node $v\in G_2 \setminus G_{bad}$ in a bucket $h(v)$, in non-increasing order of degree (and decreasing order of ID for nodes of equal degree), performed on $node_{h(v)}$, $h(v)\gets h^*$, where $h^*\sqsupseteq h(v)$ is a child bucket of $h(v)$ with $d_{h^*}^+(v)<p_{h^*}(v)$.
}
Color each node $v \in G_2 \setminus G_{bad}$ with the only palette color in its bucket.

Update the palettes of $G_{bad}$, collect to a single node, and color sequentially.
\end{algorithm}

\paragraph{Overview of coloring good nodes:} 
To color the (good) nodes in $G_2 \setminus G_{bad}$, we proceed in the following fashion: firstly, the hash function specifying the bucket for each color and the initial bucket for each node, $h()$, is chosen via the method of conditional expectations to ensure that the graph of bad nodes (by \Cref{defi:bad}) is of size $O(n)$. The algorithm then goes through iterations, and the bucket status of the nodes change over iterations (i.e. we will change the buckets $h(v)$ of nodes from the original one determined by the hash function).
 
In every iteration, for every node $v$, we restrict the color palettes of $v$ to the colors present in the descendant bucket of (current) $h(v)$. Then, for every binary string $x$ such that the bucket $x$ has at least one node, we gather the induced graph of edges with one endpoint in bucket $x$ and the other endpoint in some descendant bucket of $x$ (i.e., $G^+_{x}$) to a designated network node in $O(1)$ rounds. Note that this graph is of size $O(n)$. Though all graphs $G_x^+$ are valid \DILC instances in every iteration, the $G^+_x$ are not necessarily independent: there can be an edge between a node $v$ in $G^+_x$ and a node $w$ outside $G^+_x$ such that the current palettes of $v$ intersects with the current palette of $w$. We can show that every node $v$ satisfies $d_{h(v)}^+(v)<p_{h(v)}(v)$ in the first iteration --- i.e., each node has enough colors in its bucket to be greedily colored in  non-increasing order of degree (and decreasing order of ID between same-degree nodes). That is, (in the first iteration) each graph $G_x^+$ is a valid \DILC instance.
 
In each iteration, for every bucket $x$ in parallel, the designated network node $node_x$ will move each node $v$ with $h(v)=x$ down to a child bucket $h^*$ of its current bucket $h(v)$, in such a way that we maintain this colorability property $d_{h^*}^+(v)<p_{h^*}(v)$ (having more colors in the palette than higher degree neighbors in descendant buckets). It does so in non-increasing order of degree (and decreasing order of ID between same-degree nodes). Due to this ordering, when $v$ is moved to a child bucket, the network node $node_{h(v)}$ performing the computation knows how many of $v$'s higher-degree (those in $N^+_{G^+_{h(v}}(v)$) are in descendants of each such child bucket. This is because, for same-level neighbors in the same bucket, $node_{h(v)}$ is also responsible for those neighbors and has placed them in child buckets first. $v$'s higher-level neighbors are already in a child bucket of $h(v)$ or descendant thereof, and their bucket at the start of this iteration was part of the information collected to $node_{h(v)}$ - $node_{h(v)}$ does not need to know which bucket they are moved to (by a different network node in parallel) \emph{this iteration}, and so this bucket-updating process can be correctly done in parallel.

This will imply that, when we find graphs $G_x^+$ in the next iteration,
those are also valid \DILC instances. We will show that after $O(1)$ iterations, each node has only $1$ palette color in its bucket (and therefore zero higher-degree neighbors in descendant buckets, since $d_{h^*}^+(v)<p_{h^*}(v)$). At this point, nodes can safely color themselves the single palette color in their bucket (see Line 12 of \Cref{alg:bucketcolor}). To decide on child buckets for nodes, it is essential that each $G^+_x$ will always fit onto a single network node (which is in fact the case).
 
To prove the correctness of \textsc{BucketColor} formally, we give \Cref{lem:bad,lem:bucket_color_20_iterations}, which jointly imply \Cref{lem:evencolor}. 

\begin{lemma}
\label{lem:bad}
All network nodes can simultaneously deterministically choose a hash function (in line 1 of \textsc{BucketColor}) such that the size of $G_{bad}$ is $O(n)$.
\end{lemma}

\begin{lemma}
\label{lem:bucket_color_20_iterations}
After $20$ iterations of the outer for-loop of \textsc{BucketColor}, all nodes in $G_2\setminus G_{bad}$ can be colored without conflicts.
\end{lemma}


\begin{lemma}
\label{lem:evencolor}
\textsc{BucketColor} successfully colors graph $G_2$ in $O(1)$ rounds.
\end{lemma}

The proof of \Cref{lem:bad} is presented in \Cref{sec:color-bad} and the proofs of \Cref{lem:bucket_color_20_iterations,lem:evencolor} are in \Cref{sec:color-good}.

\subsection{Size of $G_{bad}$}\label{sec:color-bad}

To bound the size of $G_{bad}$ in expectation, consider the following lemma that says that any particular node is bad with low probability.

\begin{lemma}\label{lem:badness}
If $h:(\mathcal{C} \cup V(G_2))\rightarrow \{0,1\}^{O(\log n)}$ is chosen uniformly at random from a $100$-wise independent hash family $\mathcal H$ of functions $[n^{O(1)}]\rightarrow [n^{O(1)}]$, then the probability that any fixed node $v$ is bad is $O(d(v)^{-2})$.
\end{lemma}

\begin{proof}
Note that a node $v$ can be bad if one of four conditions mentioned in \Cref{defi:bad} hold. We show that each of them occur with probability $O(d(v)^{-2})$. We can assume throughout that $d(v)$ is at least some sufficiently large constant $C$, since otherwise every probability is trivially $O(d(v)^{-2})$.
	
	
\paragraph{Event 1: $d^+_{h(v)}(v) \ge d^+(v) 2^{-b(l(v))} + \frac 18 d(v)^{0.9}2^{-b(l(v))}$.}
For a vertex $u \in N^+(v)$, let $X_u$ be the indicator random variable that takes value $1$ or $0$ depending on whether $u$ is present in some descendant bucket of $v$ or not, i.e., $u \in N^+_{h(v)}(v)$). So, $d^+_{h(v)}(v)=\size{N^+_{h(v)}}=\sum_{u \in N^+(v)}X_u$.
	
Observe that $\Prob{X_u=1}=\Prob{h(u) \sqsupseteq h(v)}=2^{-b(\ell(v))}$, as $h(v)$ is a binary string of length $b(l(v))$ and $h(u)$ is a binary string of length at least that of $h(v)$. Hence, the expected value of $d^+_{h(v)}(v)$ is $\mu_d= \Exp{d^+_{h(v)}(v)}= d^+(v) 2^{-b(l(v))}$. Moreover, the set of random variables $\{X_u: u\in N^+(v)\}$ are $100$-wise independent as $h$ is chosen from $100$-wise independent hash family.

By \Cref{lem:conc} (and since we can assume $d(v)$ is at least a sufficiently large constant),

\begin{align*}
    \Prob{|d^+_{h(v)}-\mu_d| \ge \frac 18 d(v)^{0.9}2^{-b(l(v))}}&\le 8\left(\frac{100\mu_d+10000}{ \frac {1}{64} d(v)^{1.8}2^{-2b(l(v))}}\right)^{50}\\
    &\le 8\left(\frac{6400 d^+(v) 2^{-b(l(v))}+640000}{  d(v)^{1.8}2^{-2b(l(v))}}\right)^{50}\\
    &\le 8\left(\frac{6400 \cdot  2^{b(l(v))}+d(v)^{-0.9}2^{2b(l(v))}}{  d(v)^{0.8}}\right)^{50}\\
    &\le 8\left(\frac{6400 \cdot  d(v)^{0.7}+d(v)^{-0.5}}{  d(v)^{0.8}}\right)^{50}&\text{using $b(l(v))\le 0.7\log d(v)$}\\
    &= O\left(d(v)^{-0.1}\right)^{50}\\
    &=O(d(v)^{-5})
    \enspace.
\end{align*}

	
\paragraph{Event 2: $p_{h(v)}(v)\le p(v) 2^{-b(l(v))}  -\frac 18 d(v)^{0.9}2^{-b(l(v))}$.}
The proof is very similar to that of Event 1. Since the probability a color $c$ of $v$'s palette is present in some descendant bucket of $v$ is $\Prob{h(c) \sqsupseteq h(v)}=2^{-b(\ell(v))}$, the expected value of $p_{h(v)}(v)$ is $\mu_p=\Exp{p_{h(v)}(v)}=p(v)  2^{-b(\ell(v))} $. We may also assume that $p(v) \leq 1.25 d(v)$, which implies $\mu_p \leq 1.25d(v) 2^{-b(\ell(v))}$, because of the following. Recall that, for each $v$ in $G_2$, either $p(v)\geq d(v)+\frac{1}{4}d(v)^{0.9}$ or $N^-(v)\leq \frac{1}{3}d(v)$. So, observe that we can remove the excess colors from the palettes of vertices of $G_2$ such that $p(v)\leq 1.25 d(v)$ for each $v$ in $G_2$, and without breaking the desired property of $G_2$, i.e., either $p(v)\geq d(v)+\frac{1}{4}d(v)^{0.9}$ or $N^-(v)\leq \frac{1}{3}d(v)$.

By \Cref{lem:conc} (and since we can again assume $d(v)$ is at least a sufficiently large constant),

\begin{align*}
    \Prob{|p_{h(v)}-\mu_p| \ge \frac 18 d(v)^{0.9}2^{-b(l(v))}}&\le 8\left(\frac{100\mu_p+10000}{ \frac {1}{64} d(v)^{1.8}2^{-2b(l(v))}}\right)^{50}\\
    &\le 8\left(\frac{6400 \cdot 1.25d(v) 2^{-b(\ell(v))}+640000}{  d(v)^{1.8}2^{-2b(l(v))}}\right)^{50}\\
    &\le 8\left(\frac{8000 \cdot  2^{b(l(v))}+d(v)^{-0.9}2^{2b(l(v))}}{  d(v)^{0.8}}\right)^{50}\\
    &\le 8\left(\frac{8000 \cdot  d(v)^{0.7}+d(v)^{-0.5}}{  d(v)^{0.8}}\right)^{50}&\text{using $b(l(v))\le 0.7\log d(v)$}\\
    &= O\left(d(v)^{-0.1}\right)^{50}\\
    &=O(d(v)^{-5})
    \enspace.
\end{align*}
	
\paragraph{Event 3: any of $v$'s level $l(v)+20$ descendant buckets have more than one of $v$'s palette colors.} 
Let us consider a color $c \in \Psi(v)$ (present in the palette of $v$) and a descendant bucket $x$ of $v$ such that the level of $x$ is $l(v)+20$. The probability that $x$ contains color $c$ is equal to $\Prob{h(c) \sqsupseteq x}=2^{-b(l(v)+20)}$, as $x$ is binary string of length $b(l(v)+20)$ and $h(c)$ is binary string of length $b(\lfloor \log _{1.1}\log _2\Delta\rfloor+ 20)\geq b(l(v)+20)$. The probability that two particular colors of $v's$ palette are present in bucket $x$ is $2^{-2b(l(v)+20)}$. This follows since  the hash function $h$ is chosen uniformly at random from a $100$-wise independent hash family.
	
Hence, the probability that a particular one of $v$'s level $l(v)+20$ descendant buckets have more than one of $v$'s palette colors is at most $\binom{p(v)}{2} 2^{-2b(l(v)+20)}$. Taking a union bound over all $2^{b(l(v)+20)}$ such descendant buckets of level $l(v)+20$, the probability of any such bucket of level $l(v)+20$ containing more than one of $v$'s colors is at most
\begin{align*}
    \binom{p(v)}{2} 2^{-2b(l(v)+20)} \cdot 2^{b(l(v)+20)}
        =
    \binom{p(v)}{2} 2^{-b(l(v)+20)}
	   <
    \frac{p(v)^2}{2}\cdot  2^{-b(\log_{1.1} \log_2 p(v)+19)}\\
	   <
    \frac{p(v)^2}{2}\cdot   2^{-0.7\cdot 1.1^{\log_{1.1} \log_2 p(v)+19}}
        <
    p(v)^2\cdot   2^{-4 \log_2 p(v)}
        =
    p(v)^{-2} \leq d(v)^{-2}
    \enspace.
\end{align*}
	
\paragraph{Event 4: bucket $h(v)$ contains more than $2n 2^{-b(l(v))}$ nodes.}
For a node $u\ne v$ of level $l(v)$, the probability of $u$ being in the bucket of $v$ is $\Prob{h(u)=h(v)}=2^{-b(l(v))}$. Therefore, the expected number of nodes present in bucket $h(v)$ is $\mu \le n2^{-b(l(v))}$.  We assume that $\mu \ge 1000$, since otherwise the subgraph of $G$ induced by all such nodes must be of size $O(n)$  and would have been collected in $O(1)$ rounds and solved on a single network node. Furthermore, the events $\{\mbox{$u$ is in bucket $h(v)$}:u \neq v~\mbox{and}~l(u)=l(v)\}$ are 100-wise independent. So, by \Cref{cor:conc}, 
\begin{align*}
    \Prob{\text{bucket $h(v)$ has at least $\mu +n 2^{-b(\ell(v))}$ nodes (excluding $v$)}}
        \le
    \left(\frac{111\mu}{n^2 2^{-2b(l(v))}}\right)^{50}\\
        \le
    \left(\frac{111 \cdot 2^{-b(l(v))}}{n}\right)^{50}
        \le
    \left(\frac{111 \cdot n^{0.7}} {n }\right)^{50}
        =
    O(n^{-15})
    \enspace.
\end{align*}
	
So, the probability that (including $v$) bucket $h(v)$ contains more than $2n2^{-b(l(v))}$ nodes is $O(n^{-15})$. 
	
Note that a node $v$ is bad if one of the four above events occurs. By the union bound over all four events,  that node $v$ is bad with probability at most $ O(d(v)^{-5}) + O(d(v)^{-5}) + d(v)^{-2} + O(n^{-15}) = O(d(v)^{-2})$.
\end{proof}

Finally, we conclude with the proof of \Cref{lem:bad}: that all nodes can choose a hash function such that the size of $G_{bad}$ is $O(n)$.

\begin{proof}[Proof of \Cref{lem:bad}]
To bound the expected size of $G_{bad}$, let us define an indicator random variable $\mathbf{1}_v$ for each node $v$ in $G_2$. $\mathbf{1}_v$ takes value $1$ if $v$ is bad, and $0$ otherwise. Let us define another random variable $I_v$ for each node $v$ in $G_2$, where $I_v=d(v)\cdot \mathbf{1}_v${. Intuitively $I_v$ counts the number of ``edge endpoints'' which would be bad, if $v$ were bad. Since edges are in $G_{bad}$ if and only if both endpoints are bad, $\size{G_{bad}} \leq \frac{1}{2} \cdot \sum_{v \in G_2} I_v$.}   

Our local cost function $f(h,v)$ will take value $I_v$ when bucket choices are determined by hash function $h$, and, as usual, the global cost function $F(h)$ is given by $\sum_{v\in G_2}f(h,v)$.

By \Cref{lem:badness}, $\Exp{\mathbf{1}_v}=O((d(v))^{-2})$, and hence $\Exp{I_v} = d(v) \cdot O(d(v)^{-2}) = O(1)$. By linearity of expectation, the expected size of $G_{bad}$ is at most $\sum_{v \in G_2} \Exp{I_v}=O(n)$. Note that $\sum_{v\in G_2}I_v$ is an aggregation of $I_v$'s. Each node $v$ can locally compute the value of $I_v$ based on the conditions in \Cref{defi:bad} if the $1$-hop neighborhood of a node $v$ is known, which is the case in \CONGESTEDC. The method of conditional expectation applied to $F(h)$ then implies that we can find a hash function deterministically  such that $\sum_{v\in G_2}I_v=O(n)$, i.e.,   the size of the subgraph induced by the bad nodes is at most $O(n)$.
%
\end{proof}

\subsection{Coloring Good Nodes}
\label{sec:color-good}

Here, we prove that \textsc{BucketColor} colors the (good) nodes in $G_2 \setminus G_{bad}$ in $O(1)$ rounds. Note that \textsc{BucketColor} iteratively colors the (good) nodes in \( G_2 \setminus G_{bad} \) by executing lines 3 to 11 of \Cref{alg:bucketcolor} in each iteration. The following lemma is the main technical lemma that proves various invariants over the iterations of \textsc{BucketColor}.

\begin{lemma}\label{lem:inetr-bucketcolor}
 In every iteration $r \in \mathbb{N}$,
\begin{enumerate}
\item[(A)] For every bucket $x$, $G_x^+$ is of $O(n)$ total size;
\item[(B)] For each node $v$, $d^+_{h(v)}(v)<p_{h(v)}(v)$;
\item[(C)] When a node $v$ is considered to be moved to a child bucket, there will be at least one such child bucket $h^*$ of $h(v)$ with $d_{h^*}^+(v)<p_{h^*}(v)$. Moreover, we can find such a child buckets for all the nodes in $O(1)$ rounds.
\end{enumerate}
\end{lemma}

To prove the above lemma, let us list all the properties of  each (good) node $v$ in $G_2 \setminus G_{bad}$:
\begin{enumerate}
\item[(1)] $p(v) \ge  d(v)+\frac 14 d(v)^{0.9} $ or $|N^-(v)|\ge \frac 13 d(v)$, i.e., $d^+(v) \leq p(v)-\frac 14 d(v)^{0.9}$;
\item[(2)] $d^+_{h(v)}(v) < d^+(v) 2^{-b(l(v))} + \frac 18 d(v)^{0.9}2^{-b(l(v))} $;
\item[(3)] $p_{h(v)}(v)> p(v) 2^{-b(l(v))}  -\frac 18 d(v)^{0.9}2^{-b(l(v))} $ (i.e., $p_{h(v)}(v)> d^+_{h(v)}(v)$);
\item[(4)] None of $v$'s level $l(v)+20$ descendant buckets contain more than one of $v$'s palette colors;
\item[(5)] At most $2n2^{-b(l(v))}$ nodes $v'$ have $h(v)=h(v')$.
\end{enumerate}
Note that Property (1) holds as nodes are in $G_2$ (see \Cref{lem:properties}) and other properties hold as nodes are good (see \Cref{defi:bad}).
\begin{proof}[Proof of \Cref{lem:inetr-bucketcolor}]
We prove \Cref{lem:inetr-bucketcolor} by induction on $r$. 
	
\paragraph{Base Case ($r=1$):}
In the first iteration, all nodes are good nodes. 
	
\paragraph*{1.} 
By Properties (2) and (5), the size of $G^+_x$ can be bounded as follows:

\begin{align*}
	\!\!\!\!\!\!\sum_{v : h(v)=x}\!\!\!d^+_{h(v)}(v) &\le 
	2n2^{-b(l(v))} \cdot \left(d^+(v) 2^{-b(l(v))} + \tfrac 18 d(v)^{0.9}2^{-b(l(v))}\right) \le 
	3n d(v) 2^{-2b(l(v))} = O(n)
	\enspace.  
\end{align*}
	
\paragraph*{2.} 
 As every good node $v$ satisfies Properties (1), (2), and (3), $d^+_{h(v)}(v)$ can be bounded as follows:
\begin{align*}
	d^+_{h(v)}(v) &< 
	d^+(v) 2^{-b(l(v))} + \tfrac 18 d(v)^{0.9}2^{-b(l(v))} \le 
	p(v) 2^{-b(l(v))} - \tfrac 18 d(v)^{0.9}2^{-b(l(v))} <
	p_{h(v)}(v)
	\enspace.
\end{align*}
	
\paragraph*{3.} 
Consider the point at which a node $v$ is considered to be moved to a child bucket. It has two types of higher-degree neighbors: those in the same level, which are also the responsibility of $node_{h(v)}$ and have already been moved to a child bucket, and those in a higher level, which are already in a strict descendant bucket of $h(v)$. (We may always discard edges whose endpoints are not in ancestor-descendant bucket pairs, since their palettes will be entirely disjoint and so they will never cause a coloring conflict.) So, all of $v$'s higher degree neighbors are in strict descendant buckets, and $node_{h(v)}$ knows this bucket for each neighbor of $v$ (the current bucket for same-level neighbors and the previous bucket for higher-level neighbors).
	
For each child bucket $B$ of $h(v)$, let $N_B^+(v)$ be the set of
\begin{itemize}
\item same-level higher-degree neighbors of $v$ which were moved to $B$ in this iteration, and
\item higher-level higher-degree neighbors of $v$ which were in $B$ at the start of this iteration.
\end{itemize}
 For each child bucket $B$ of $h(v)$, let $p_B(v)$ denote the number of $v$'s palette colors that are in $B$. Denoting $|N_B^+(v)|=d_B^+(v)$, we get 
 
 $$\sum\limits_{\text{child }B\text{ of }h(v)}d_B^+(v) = d^+_{h(v)}(v)<p_{h(v)}(v)=\sum\limits_{\text{child }B\text{ of }h(v)}p_B(v).$$ 

So, there must be at least one child bucket $h^*$ with $d^+_{h^*}(v)<p_{h^*}(v)$, and $node_{h(v)}$ can identify such a bucket. (Concurrently, the higher-level higher-degree neighbors of $v$ are moved to a child bucket, but this does not affect $d^+_{h^*}(v)$.) This is possible  in $O(1)$ rounds as $G_{h(v)}^+$ (By (1) of this lemma) is of size $O(n)$ and is gathered at $node_{h(v)}$ by spending $O(1)$ rounds.

\paragraph{Inductive step:}
Consider iteration $r+1$.

\paragraph{1.}
For a bucket $x$, let bucket $x'$ be its parent. The nodes present in bucket $x$ are a subset of the nodes present in the bucket $x'$ in iteration $r$. This is because, by the induction hypothesis, every node in bucket $x$ in iteration $r+1$ was moved from  the bucket $x'$ to bucket $x$ in iteration $r$. Also, the nodes present in bucket $x$ in iteration $r$ are now in some child bucket of $x$. Furthermore, as every descendant bucket of $x$ is also a descendant bucket of $x'$, every edge present in $G^+_x$ in iteration $r+1$ was also present in $G_{x'}^+$ in iteration $r$. So, $G_x^+$ in iteration $r+1$ is a subgraph of $G_{x'}^+$ in iteration $r$ (which is of size $O(n)$ by the induction hypothesis). 
	
\paragraph{2.}
For node $v$, let $x'$ denote the bucket of $v$ in iteration $r+1$. Consider the point we consider moving node $v$ to a child bucket of $x'$ in iteration $r$. By the induction hypothesis, we have found a child bucket $h^*$ of $x'$ such that $d^{+}_{h^*}(v)<p_{h^*}(v)$ and have moved vertex $v$ to bucket $h^*$ from $x'$ (i.e., $h(v)$ was set to $h^*$). So,  $d^+_{h(v)}(v)<p_{h(v)}(v)$ holds in iteration $r+1$. 
	
\paragraph{3.}
Using the same argument as in the base case, we can argue that there will be at least one such child bucket $h^*$ of $h(v)$ such that $d_{h^*}^+(v)<p_{h^*}(v)$. Moreover, we can  find such child buckets for all the nodes in $O(1)$ rounds.
\end{proof}

Next, we use the properties shown in \Cref{lem:inetr-bucketcolor} to show that after $O(1)$ iterations of the outer for-loop of \textsc{BucketColor}, nodes can color themselves without conflicts.

\begin{proof}[Proof of \Cref{lem:bucket_color_20_iterations}]
Consider the situation after $20$ iterations of the outer for-loop of \textsc{BucketColor}. By \Cref{lem:inetr-bucketcolor} (2) along with the fact that we move nodes to a child bucket in every iteration, each node $v$ is in a level $l(v)+20$ bucket $h^*$ satisfying $d^+_{h^*}(v)<p_{h^*}(v)$. Observe that, by Property (4) of good nodes, we must  have $d^+_{h^*}(v)=0$ and $p_{h^*}(v)=1$. Removing edges which are not between ancestor-descendant bucket pairs (since their endpoints have disjoint palettes), the remaining graph has no edge. Therefore, each node can be colored the only remaining color in its palette with no conflicts.
\end{proof}

Finally, we conclude with the proof of correctness of \textsc{BucketColor}.

\begin{proof}[Proof of \Cref{lem:evencolor}]
By \Cref{lem:bad} nodes can determine $G_{bad}$ in $O(1)$ rounds, and this graph can later be colored in $O(1)$ by a single network node.

By \Cref{lem:bucket_color_20_iterations}, after $20$ iterations of the outer for-loop of \textsc{BucketColor}, all nodes in $G_2 \setminus G_{bad}$ can be colored. By \Cref{lem:inetr-bucketcolor}~(3), each of these iterations takes $O(1)$ rounds.
\end{proof}

\section{Proof of the Main Theorem}
\label{sec:pf-main}

Now, we are ready to complete our analysis of a constant-round \CONGESTEDC and prove \Cref{thm:D1LCcolor}. We begin with a theorem summarizing the properties of \textsc{Color$(G,0)$}.

\begin{theorem}
\label{thm:analysis-Color}
\textsc{Color$(G,0)$} colors any \DILC instance $G$ with $\Delta_G \le O(\sqrt{n})$ in $O(1)$ rounds.
\end{theorem}

\begin{proof}
From the description of \textsc{Color$(G,0)$} and its subroutines, it is evident that \textsc{Color$(G,0)$} colors a graph $G$ successfully when $\Delta_G \leq O(\sqrt{n})$. It remains to analyze the total number of rounds spent by \textsc{Color$(G,0)$}.

Note that the steps of \textsc{Color$(G,0)$}, other than the call to subroutines \textsc{ColorTrial}, \textsc{SubSample}, \textsc{BucketColor} and recursive call, can be executed in $O(1)$ rounds. \textsc{ColorTrial} and \textsc{SubSample} can be executed in $O(1)$ rounds by \Cref{lem:sizeF} and \Cref{lem:G'size}, respectively. Also, $O(1)$ rounds are sufficient for \textsc{BucketColor}
 due to \Cref{lem:evencolor}.

To analyze the round complexity of recursive calls in \textsc{Color$(G,0)$}, let $G^i$ denote the graph on which the $i^{th}$-level recursive call of \textsc{Color}, i.e., \textsc{Color$(G^i,0.1i)$} is made. While executing \textsc{Color$(G^i,0.1i)$}, the algorithm spends $O(1)$ rounds  and makes a recursive call \textsc{Color$(G^{i+1},0.1(i+1))$}.

We show by induction that $\sum_{\text{$v\in G_{i}$}}d_{G_{i}}(v)^{0.1i} \le 3^i Cn$ for $i\le 10$. This is true for $G^0 = G$, since $\sum_{\text{$v\in G$}}d_{G}(v)^{0} = n$.
For the inductive step, for $0 \le i \le 9$, by \Cref{lem:G'size} using $x=0.1i$,
\begin{align*}
    \mbox{}\!\!\!\!\!\!\!\!\sum_{v\in G^{i+1}}\!d_{G^{i+1}}(v)^{0.1(i+1)}
        \le
    \!\!\!\sum_{v\in G^{i+1}}\!d_{G^{i}}(v)^{0.1(i+1)}
        \le
    Cn+2\!\sum_{v\in G^{i}}\!d_{G^i}(v)^{0.1i}
        \le
    Cn+2\cdot 3^i Cn 
        \le
    3^{i+1} Cn
    \enspace.
\end{align*}

So, $|E(G^{10})| \le \sum_{\text{$v\in G^{10}$}}d_{G^{10}} \le 3^{10} Cn = O(n)$. Therefore, after $10$ recursive calls, the remaining uncolored graph can simply be collected to a single network node and colored.
\end{proof}

While \Cref{thm:analysis-Color} requires that $\Delta_G \leq O(\sqrt{n})$, we note that we can generalize this result to arbitrary maximum degree.

\begin{lemma}
\label{lemma:reduce-to-sqrt}
In $O(1)$ rounds of \CONGESTEDC, we can recursively partition an input \DILC instance into sub-instances of maximum degree $O(\sqrt{n})$. The sub-instances can be grouped into $O(1)$ groups where each group can be colored in parallel.
\end{lemma}

\begin{proof}
We use the \textsc{LowSpacePartition} procedure from \cite{CDP21a}. The process involves two hash functions to organize nodes and colors into bins: the first maps nodes to \( n^\zeta \) bins, creating \( n^\zeta \) graphs; the second restricts the color palettes of nodes in the first \( n^{\zeta - 1} \) bins, where $\zeta \in (0,1)$ is a suitable constant. The resulting graphs are then returned. The recursive use of \textsc{LowSpacePartition} reduces a coloring instance to $O(1)$ sequential instances of maximum degree $n^\eps$ for any constant $\eps>0$. The procedure is for $\Delta+1$-coloring, but it extends immediately to \DILC, as discussed in Section~5 of \cite{CCDM23}. Since we can simulate low-space \MPC in \CONGESTEDC, we can execute \textsc{LowSpacePartition}, setting $\eps$ appropriately to reduce the maximum degree of all instances to $O(\sqrt{n})$. By subsequent arguments in \cite{CDP21a}, the sub-instances can be partitioned into $O(1)$ groups, where each group can be colored in parallel, while the \DILC instances corresponding to each group must be solved sequentially.
\end{proof}

Now the proof of \Cref{thm:D1LCcolor} follows immediately from \Cref{thm:analysis-Color} and \Cref{lemma:reduce-to-sqrt}.
\qed